\pdfoutput=1

\documentclass[journal]{IEEEtran}
\usepackage{bm}

\usepackage{amsmath}
\usepackage{amsthm}
\usepackage{makecell}
\usepackage{multirow}
\usepackage{subcaption}
\newtheorem{definition}{Definition}
\newtheorem{theorem}{Theorem}

\ifCLASSINFOpdf
  \usepackage[pdftex]{graphicx}
\else
\fi
\usepackage{algorithm}
\usepackage{algorithmic}
\begin{document}
%
\title{Game-based Pricing and Task Offloading in Mobile Edge Computing Enabled Edge-Cloud Systems}
%
%
%

\author{Yi Su,
	Wenhao Fan, \IEEEmembership{Member, IEEE},
	Yuan'an Liu, \IEEEmembership{Member, IEEE},
	and Fan Wu
\thanks{Yi Su, Wenhao Fan, Yuan'an Liu and Fan Wu are with the School of Electronic Engineering, and Beijing Key Laboratory of Work Safety Intelligent Monitoring, Beijing University of Posts and Telecommunications, Beijing, China (e-mail: whfan@bupt.edu.cn, yuliu@bupt.edu.cn).}}

\maketitle

\begin{abstract}
As a momentous enabling of the Internet of things (IoT), mobile edge computing (MEC) provides IoT mobile devices (MD) with powerful external computing and storage resources. However, a mechanism addressing distributed task offloading and price competition for the open exchange marketplace has not been established properly, which has become a huge obstacle to MEC's application in the IoT market. In this paper, we formulate a distributed mechanism to analyze the interaction between OSPs and IoT MDs in the MEC enabled edge-cloud system by appling multi-leader multi-follower two-tier Stackelberg game theory. We first prove the existence of the Stackelberg equilibrium, and then we propose two distributed algorithms, namely iterative proximal offloading algorithm (IPOA) and iterative Stackelberg game pricing algorithm (ISPA). The IPOA solves the follower non-cooperative game among IoT MDs and ISPA uses backward induction to deal with the price competition among OSPs. Experimental results show that IPOA can markedly reduce the disutility of IoT MDs compared with other traditional task offloading schemes and the price of anarchy is always less than 150\%. Besides, results also demonstrate that ISPA is reliable in boosting the revenue of OSPs.
\end{abstract}

\begin{IEEEkeywords}
Internet of things, mobile edge computing, Stackelberg game
\end{IEEEkeywords}

%
\IEEEpeerreviewmaketitle

\section{Introduction} \label{sec-1}
\IEEEPARstart{W}{ith} the rapid development and popularization of Internet of Things (IoT) technology, the number of deployed IoT devices is experiencing explosive growth \cite{lin2017survey}. It is estimated that the number of IoT devices will reach an astonishing 21.5 billion by 2025 \cite{Lueth2018iotnumber}. These devices will generate massive data and demand further processing, providing intelligence to both service providers and users \cite{yu2017survey}. However, most IoT devices, especially mobile devices (MDs) have highly constrained computing power and battery capacity, which means that it is unrealistic to meet the demands of IoT applications by processing all the raw data locally. Although traditional cloud computing allows IoT MDs to offload their computation tasks to the remote data center so as to utilize powerful central processing unit (CPU) and ample storage capabilities of cloud \cite{othman2013survey}, considerable transmission costs are incurred because the data centers that cloud computing relies on are geographically far away from MDs. In addition to this, offloading tasks to the cloud imposes huge additional burdens on the backbone network traffic load and the performance of the network will worsen with increasing data size.

In order to overcome the challenges associated with centralized cloud computing environments, the European Telecommunications Standards Institute (ETSI) introduced the term Mobile Edge Computing (MEC) in 2014 and further dropped the 'Mobile' out of MEC and renamed it as Multi-access Edge Computing in order to broaden its applicability into heterogeneous networks in 2016 \cite{taleb2017multi}. MEC is initiated aiming to create an open environment across multi-vendor cloud platforms located at the edge of access network, accessible by application/service providers and third parties \cite{hu2015mobile}.
Seen as a key technology for 5G wireless systems, the servers of MEC can be located at the base stations (BSs) in a fully distributed manner, enabling the delivery of locally-relevant, fast services. Offloading tasks to the MEC servers at BSs will bring markedly diminishing transmission delay compared with mobile cloud offloading. With the help of MEC, sensitive data can also be confined to local zones and not exposed to the internet, enhancing IoT information security. As lots of literature advocate \cite{ren2019survey,guo2018collaborative,ebrahimzadeh2020cooperative}, MEC is not a substitute but a complement to cloud computing and a cooperative edge-cloud system can provide IoT MDs with diverse offloading choices to improve their quality of experience (QoE). 

Although mobile computation offloading in MEC has been extensively studied in the literature \cite{guo2018collaborative,ebrahimzadeh2020cooperative,mach2017mobile,yi2019multi,8241344}, there are still many challenging issues with integrating MEC to assist IoT: (\romannumeral1) IoT MDs are more sensitive to delay and energy consumption and the computation tasks of different IoT MDs vary greatly, all of which make it difficult to model the utility function of IoT MDs; (\romannumeral2) Due to IoT connecting a diverse assortment of devices belonging to different parties, the service resources may not belong to a single offloading service provider (OSPs) \cite{yu2017survey}; (\romannumeral3) Selfish OSPs and IoT MDs are interested in optimizing their own utility individually in the collaborative edge-cloud MEC system, which further increases the difficulty of offloading strategy optimization and system stability. Many studies leverage auction theory to assign computing resources and design pricing policies between multiple OSPs and IoT MDs \cite{7296648,8744396,zhang2019near}. However, in those studies, there is always a trustworthy third party acting as an auctioneer, which may not be found in some IoT scenarios. Therefore, designing a distributed computing offloading and pricing mechanism in such an authority-lacking competitive edge-cloud system is still an urgent problem to be solved, which interests us in conducting an in-depth study of this problem.

In this paper, we are concentrating on designing a novel distributed task offloading and pricing mechanism in an edge-cloud enabled IoT environment, where multiple edge or cloud OSPs provide distinct offloading services and IoT MDs will offload their computation tasks to different edge servers and cloud servers proportionally, according to the prices announced by OSPs. To address the urgent problem of task offloading and pricing in a fully distributed manner, we place the trading between OSPs and IoT MDs in an open exchange marketplace and the mechanism based on economic principles is applied in such marketplace. OSPs are seen as leaders and IoT MDs are considered as followers because the offloading strategies of IoT MDs are determined after OSPs' prices are given. The interaction between OSPs and IoT MDs is formulated as a multi-leader multi-follower two-tier Stackelberg game \cite{leyffer2005solving} and Stackelberg equilibrium (SE) exists in our proposed mechanism. Our main contributions include:

\begin{enumerate}
	\item \textit{Modeling}. A novel disutility function is formulated for IoT MDs to quantify their QoE. Comprehensively considering queuing models at different stages and the individual differences in the importance of delay, energy consumption and payment, this model can accurately reflect how the pricing of OSPs and the offloading strategies of IoT MDs affect the QoE of IoT MDs.
	\item \textit{Game analysis}. The interaction between OSPs and IoT MDs is regarded as a multi-leader multi-follower two-tier Stackelberg game, which is composed of a leader non-cooperative game for OSPs and a follower non-cooperative game for IoT MDs. We employ the variational inequality (VI) approach to analyze the existence and uniqueness of Nash equilibrium (NE) in the follower non-cooperative game given the prices of OSPs. For the leader non-cooperative game, the existence of NE is also derived, verifying the existence of Stackelberg equilibrium (SE).
	\item \textit{Algorithms design}. A distributed iterative proximal offloading algorithm (IPOA) is proposed to address the problem of task offloading given the prices of all OSPs. This algorithm can converge to an NE in limited iterations. Additionally, we propose another iteritive Stackelberg game pricing algorithm (ISPA) to solve the leader non-cooperative pricing game among OSPs by applying backward induction. 
	\item \textit{Performance evaluation}. Simulations are conducted to investigate the performance of our proposed mechanism. Results show that IPOA can markedly improve IoT MDs’ QoE compared with other traditional task offloading schemes. The price of anarchy (PoA), which reflects the gap in overall performance between NE and socially optimal offloading, is bounded. We further study the ISPA and prove that OSPs can find appropriate prices for all OSPs.	     
\end{enumerate}

The rest of the paper is organized as follows. In Sec. \ref{sec-2}, we introduce related works. The system model is described in Sec. \ref{sec-3}. Based on the system model, we analyze the Stackelberg game between OSPs and IoT MDs in Sec. \ref{sec-4}. Furthermore, in Sec. \ref{sec-5}, two algorithms are proposed to deal with distributed task offloading and pricing problems. Sec. \ref{sec-6} investigates the performance of our proposed algorithms. We conclude the paper in Sec. \ref{sec-7}.

 




\section{Related Works} \label{sec-2}

The computation offloading, and partial offloading in particular, is a very complex process affected by different factors, such as users preferences, radio and backhaul connection quality and cloud capabilities \cite{othman2013survey}. Considering different factors or adopting different methods, many literatures have launched studies into mobile computing offloading in MEC. Fan et al. \cite{8241344} considered the problem of excessive load on a specific MEC-BS, and the scheme they proposed could not optimize the performance of the entire system. In \cite{ebrahimzadeh2020cooperative}, the authors investigated the performance gains of cooperative edge-cloud computation offloading for MEC enabled FiWi enhanced HetNets and presented a self-organization based mechanism to enable mobile users. Liu et all. \cite{liu2017multiobjective} brought a thorough study on the energy consumption, execution delay and payment cost of offloading processes in a fog computing system where only one OSP exists. They proposed an algorithm aiming to minimize the average computing cost of all MDs in a central manner, which did not apply to competitive environments.
The authors of \cite{chen2018optimized} considered the task queue state, the energy queue state as well as the channel qualities between MU and BSs in the MEC systems and modeled the optimal computation offloading problem as a Markov decision process. Deep reinforcement learning based algorithms were proposed to optimize task computation experience of mobile users. Although there are many approaches using deep reinforcement learning to design computation offloading policies in the MEC, even some works like \cite{chen2019multi} applies deep reinforcenment learning to solve the competition among OSPs, we have not yet seen the possibility of applying deep reinforcement learning based optimization methods to address computation offloading and pricing jointly in distributed edge-cloud systems.                  

As a well-researched economic theory, auction theory has been employed by many literatures to allocate resources from edges to devices in non-cooperative scenarios. In \cite{7296648}, Jin et al. proposed an incentive-compatible auction mechanism (ICAM) based on the single-round double auction model for the cloudlet scenario. The ICAM can effectively allocate cloudlets (sellers) to satisfy the service demands of mobile devices (buyers) and determine the pricing. Wang et al. designed an online profit maximization multi-round auction (PMMRA) mechanism for the computational resource trading between edge clouds and mobile devices in a competitive MEC environment and an outperforming profit of edge clouds was obtained in \cite{8744396}. In \cite{zhang2019near}, energy harvesting-enabled MDs were considered as offloading service providers and the offloaded tasks generated by IoT devices were optimally assigned through a proposed online rewards-optimal auction (RoA). Despite the focus on the resource allocation on the provider side, these studies ignore the impact of MD's offloading strategy.

Stackelberg game model is a classic model in game theory and several works advocate Stackelberg game theory as an effective solution concept for pricing and resource allocating problems in the game between resource providers and consumers \cite{zhang2017computing, wang2017multi, 8478374, jie2020game}. In \cite{8478374}, the complicated interactions among unmanned aerial vehicles (UAVs) and BSs as well as the cyclic dependency was considered as a Stackelberg game where BSs were leaders determining the bandwidth allocated to each UAV and UAVs acted as followers to select the bandwidth. The authors solved the Stackelberg game through backward induction as the UAV payoff information was requested by BSs. Jie et al. \cite{jie2020game} utilized the double-stage Stackelberg game to propose an optimal resource allocation scheme between cloud center (CC) and data users (DUs) by introducing fog service providers (FSPs) for a fog-based industrial internet of things (IIoT) environment. They modelled the competition of FSPs as a non-cooperative game in respect of the fact that FSPs competed to buy resources from the CC in order to provide paid services to DUs. However, they did not consider the competition between DUs assuming the resources needed by each DU are distinct. Three algorithms were designed and Nash equilibrium and Stackelberg equilibrium were achieved in the end.

To the best of our knowledge, there are still no solutions to address task offloading and pricing jointly in a fully distributed manner for the competitive heterogeneous edge-cloud systems, where not only the utility information of each IoT MD is private but also the pricing of an OSP will not be known in advance by other OSPs before it is released. Accurately, how to achieve the equilibrium of such a system is still a challenging problem.                   

\section{System Models} \label{sec-3}

\begin{figure}[!t]
	\centering
	\includegraphics[width=3in]{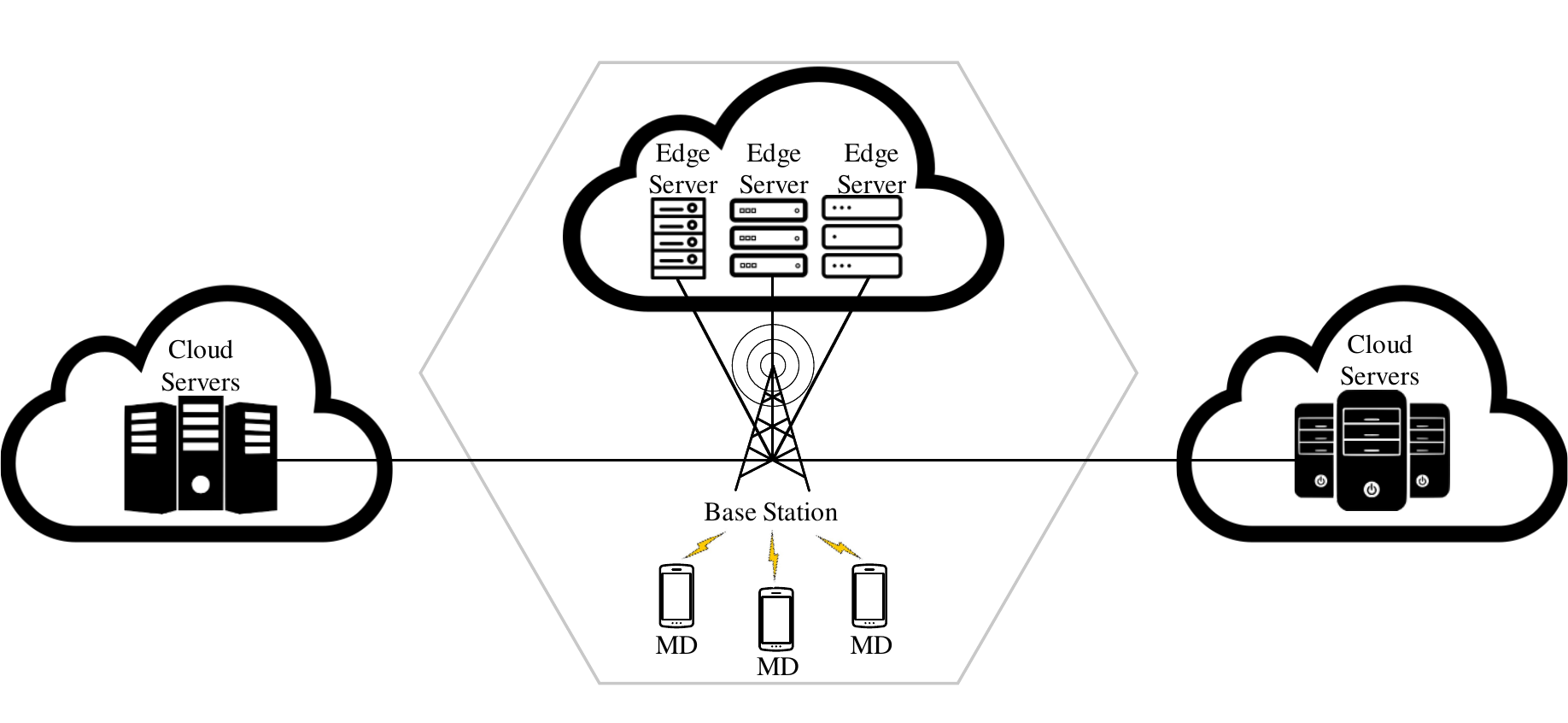}
	\DeclareGraphicsExtensions.
	\caption{System architecture}
	\label{fig:system_architecture}
\end{figure}

As Fig. \ref{fig:system_architecture} depicts, we consider there is a set of $ M $ IoT MDs, which is denoted as ${\cal M} = \left\{ {1,...,M} \right\}$, within the coverage of a BS and it is assumed that the $ M $ IoT MDs can can only access the same BS at the same time. All of these IoT MDs can access to the BS simultaneously. There are $ N $ OSPs in the system, which consist of $ N_e $ edge computing OSPs and $ N_c $ cloud computing OSPs. Each edge computing OSP deploys an edge server at the BS and each cloud computing OSP owns a cloud computation center connected to BS via optical backbone network \cite{ebrahimzadeh2020cooperative}. We define the set of OSPs as $ {\cal N}=\left\{1,...,N\right\} $, of which the first $ N_c $ elements represent cloud computing OSPs and the rest are edge computing OSPs. The prices announced by OSPs are broadcast via BS periodically. After receiving the prices, IoT MDs determine their offloading strategies and send their offloading requests to the BS. Compared with task input data, the communication data used for publishing prices and notifying offloading requests is much smaller, which ensures the feasibility of distributed algorithms.

In this paper, for a given IoT MD $ i $, who is involved in ${\cal M}$, we assume it only generates independently offloadable tasks following a Poisson process with an average rate $ \lambda_{i} $ same as many papers \cite{ebrahimzadeh2020cooperative,liu2017multiobjective,shah2017joint}. Each task of IoT MD $ i $ is characterized by $ c_{i} $ and $ z_{i} $, which denote the average number of CPU cycles required and the average size of computation input data (e.g., program codes and input parameters) \cite{ebrahimzadeh2020cooperative}. The generated tasks can not only be executed locally but can also be partially offloaded to OSPs' servers. In other words, IoT MD $ i $ can offload each of its tasks to any OSPs. We also assume that a task is offloaded from IoT MD $ i $ to OSP $ j $ with a probability of $ \alpha_{i,j} $ (or equivalently the long-term offloading ratio \cite{shah2017joint}). Therefore, for MD $ i $, there is an offloading strategy vector represented as
\begin{equation*}
{\bm{ \alpha }_{i}} = {\left( {{\alpha _{i,1}},...,{\alpha _{i,N}}} \right)^T}.
\end{equation*}
Obviously, vector $ \bm{ \alpha }_{i} $ is subjected to the constraint ${\alpha _{i,j}} \in \left[ {0,1} \right]$ and ${\rm{0}} \le \sum\nolimits_{j \in {\cal N}} {{\alpha _{i,j}} \le 1}$.

We summarize all the key notations in Table. \ref{table1}. What needs to be emphasized is, for simplicity, in the rest of the paper, when a task is offloaded to a server belonging to OSP $ j $, we will directly say that the task is offloaded to OSP $ j $.  

\begin{table}
	\caption{Notations}
	\label{table1}
	\begin{center}
	\setlength{\tabcolsep}{2pt}
	\begin{tabular}{c|p{150px}}
		\hline
		Notation                            & Description\\
		\hline
		$M$                                 & the number of IoT MDs\\
		\hline
		${\cal M}$                          & the set of IoT MDs\\
		\hline
		${N}$                               & the number of OSPs\\
		\hline
		$N_e$                               & the number of edge computing OSPs\\
		\hline
		$N_c$                               & the number of cloud computing OSPs\\
		\hline
		${\cal N}$                          & the set of OSPs\\
		\hline
		${\lambda _i}$                      & the average task arrival rate of IoT MD $i$\\
		\hline
		\multirow{2}{*}{\centering $ c_i $} & the average number of CPU cycles required \\
							  			    & by each task of IoT MD $ i $ \\
		\hline
		\multirow{2}{*}{\centering $ z_i $} & the average size of computation input data \\
											& of each task of IoT MD $ i $ \\
		\hline
		\multirow{2}{*}{$ \alpha _{i,j} $}  & the probability that a task is offloaded from \\
			                                & IoT MD $i$ to OSP $j$\\
		\hline
		$ \bm{\alpha} _i $                  & the offloading strategy vector of IoT MD $ i $\\
		\hline
		$ \bm{A} $                          & the offloading strategies profile of all IoT MDs\\
		\hline
		\multirow{2}{*}{$ \bm{A}_{-i} $}    & the offloading strategies profile of all IoT MDs\\
		                                    & except $ i $\\
		\hline
		$ f_i^{MD} $                        & the processing capability of IoT MD $ i $ \\
		\hline
		$\varepsilon _i^{local}$            & the local computing power of IoT MD $ i $\\
		\hline
		$ \varepsilon _i^{tx} $             & the transmission power of IoT MD $ i $\\
		\hline
		$ B $                               & the wireless channel bandwidth\\
		\hline
		$ w_0 $                             & the background interference power\\
		\hline
		$ h_i $                             & the channel gain between IoT MD $ i $ and BS\\
		\hline
		\multirow{2}{*}{$ \sigma_i^{2} $}   & the variance of service times in the IoT MD \\
		                                    & $ i $'s wireless interface \\
		\hline
		\multirow{2}{*}{$ a_j $}            & the number of optical amplifying from the\\                               & BS to the cloud computation center\\
		\hline
		$ R $                               & the uplink data rate of the optical fiber network\\
		\hline
		\multirow{2}{*}{$ t $}              & the uplink propagation delay of the  optical\\                           & backbone network\\
		\hline
		\multirow{2}{*}{$ f_j^{OSP} $}      & the computing capability of MEC server\\
		                                    & owned by OSP $ j $\\
		\hline
		$ p_j $                             & the price announced by OSP $ j $\\
		\hline
		$ \bm{p} $                          & the price vector of all OSPs\\
		\hline
		$ \bm{p}_{-j} $                     & the prices announced by all OSPs except $ j $\\
		\hline
		\multirow{2}{*}{$ D_i^{MAX} $}      & the maximum computing delay IoT MD $ i $ \\
		                                    & can accept\\
		\hline
		\multirow{2}{*}{$ E_i^{MAX} $}      & the maximum enrgy consumption IoT MD $ i $ \\
							                & can accept\\
		\hline
		\multirow{2}{*}{$ D_i^{MAX} $}      & the maximum payment cost IoT MD $ i $\\
		                                    & can accept\\
		\hline
		$ \theta _i^{\rm D} $               & the weight factor of delay for IoT MD $ i $\\
		\hline
		$ \theta _i^{\rm E} $               & the weight factor of energy for IoT MD $ i $\\
		\hline
		$ \theta _i^{\rm P} $               & the weight factor of payment for IoT MD $ i $\\
		\hline
		$ U_i^{MD} $                        & the disutility function of IoT MD $ i $\\
		\hline
		$ U_j^{OSP}$                        & the utility function of OSP $ j $\\
		\hline
	\end{tabular}
	\end{center}
\end{table}

\subsection{Local Computing}

Due to the limited computing capability of an IoT MD, we consider there is an $ M/M/1 $ queue model in the local CPU with tasks arriving rate $\left( {1 - \sum\nolimits_{j \in {\cal N}} {{\alpha _{i,j}}} } \right){\lambda _i}$ for every IoT MD $i$ same as \cite{liu2017multiobjective}. Assuming the processing capability of IoT MD $ i $ is $  f_i^{MD}$, the mean delay caused by computing a task locally in IoT MD $i$ is
\begin{equation}
	\label{equ1}
	D_i^{local} = \frac{c_{i}}{{ f_i^{MD} - \left( {1 - \sum\nolimits_{j \in {\cal N}} {{\alpha _{i,j}}} } \right){\lambda _i}{c_{i}}}}.
\end{equation}

We denote the local computing power of IoT MD $i$ as $\varepsilon _i^{local}$, and the energy consumption of computing a task in IoT MD $i$ is
\begin{equation}
	\label{equ2}
	E_i^{local} = \frac{{\varepsilon _i^{local}}{c_{i}}}{{f_i^{MD} - \left( {1 - \sum\nolimits_{j \in {\cal N}} {{\alpha _{i,j}}} } \right){\lambda _i}{c_{i}}}}.
\end{equation}

In particularly, (\ref{equ1}) and (\ref{equ2}) should be subject to the condition $ \left( {1 - \sum\nolimits_{j \in {\cal N}} {{\alpha _{i,j}}} } \right){\lambda _i}{c_{i}} < f_i^{MD} $, so as to guarantee queue stability.   

\subsection{Computation Offloading}

When IoT MD $i$ transmits its tasks to BS, the uplink delay consists of the waiting time in the wireless interface, the service time of the wireless interface to submit the tasks to BS. In this paper, we assume that each IoT MD is equipped with a single antenna and the interference cause by other IoT MDs can not be ignored. According to Shannon-Hartley Theorem, the uplink data rate of IoT MD $i$ via wireless cellular network is
\begin{equation}
	\label{equ3}
	{r_i} = B{\log _2}\left( {1 + \frac{{\varepsilon _i^{tx}{h_i}}}{{{w_0} + \sum\nolimits_{k \in {\cal M},k \ne i} {\varepsilon _k^{tx}{h_k}} }}} \right),
\end{equation}
where $B$ is wireless channel bandwidth and ${\varepsilon _i^{tx}}$ is the transmission power of IoT MD $i$. ${{h_i}}$ is the channel gain between IoT MD $i$ and BS and ${w_0}$ is the background interference power. Obviously, (\ref{equ3}) considers the worst case where most channel noise is brought by all other IoT MDs. It can be noticed that as the number of IoT MDs in the network increases, the uplink data rate of each MD will decrease accordingly.

Without loss of generality, we assume that the arriving traffic at each IoT MD's wireless interface, which is the data stream that needs to be uploaded to the BS, is the input data from the offloaded tasks, which arrives following a Poisson process with arrival rate $ {\lambda_{i}}{c_{i}}{\sum\nolimits_{j \in {\cal N}} {{\alpha _{i,j}}} } $. Thus, the wireless interface queue can be modeled as an $ M/G/1 $ queuing system, and based on the Pollaczed-Khinchin formula \cite{bertsekas1992data}, the average wireless transmission delay of IoT MD $ i $ can be defined as follows:

\begin{equation}
	\label{equ4}
	{D_i^{tx,wireless}}  = \frac{{{\lambda _i}{\overline {S_i^2} }\sum\nolimits_{j \in {\cal N}} {{\alpha _{i,j}}} }}{{2\left( {1 - \frac{{{\lambda _i}{z_i}\sum\nolimits_{j \in {\cal N}} {{\alpha _{i,j}}} }}{{{r_i}}}} \right)}} + \frac{{{z_i}}}{{{r_i}}},
\end{equation}
which is subject to the condition $ {{\lambda _i}{z_i}\sum\nolimits_{j \in {\cal N}} {{\alpha _{i,j}}} } < r_{i} $. In (\ref{equ4}), $ {{{{z_i}} \mathord{\left/{\vphantom {{{z_i}} {{r_i}}}} \right.\kern-\nulldelimiterspace} {{r_i}}}} $ represents the mean service time of each offoaded task in the wireless interface. $ {\overline {S_i^2}} = \sigma _i^2 + {\left( {{{{z_i}} \mathord{\left/{\vphantom {{{z_i}} {{r_i}}}} \right.\kern-\nulldelimiterspace} {{r_i}}}} \right)^2} $ and $ \sigma _i^2 $ denotes the variance of service times. Since the downlink rate is usually much higher than the uplink rate and the size of the output results is much smaller than the input size, the download time of results is neglected as many works \cite{liu2017multiobjective,shah2017joint}. knowing that the transmission power of IoT MD $ i $ keeps fixed during input data transmission, we can get the average energy consumption for transmitting a task of IoT MD $ i $ to BS as follows:
\begin{equation}
	\label{equ5}
	{E_i^{tx}}  = \varepsilon_{i}^{tx}{D_i^{tx,wireless}} = \frac{{{\varepsilon_i^{tx}}{\lambda _i}{\overline {S_i^2} }\sum\nolimits_{j \in {\cal N}} {{\alpha _{i,j}}} }}{{2\left( {1 - \frac{{{\lambda _i}{z_i}\sum\nolimits_{j \in {\cal N}} {{\alpha _{i,j}}} }}{{{r_i}}}} \right)}} + \frac{{{\varepsilon_i^{tx}}{z_i}}}{{{r_i}}}.
\end{equation}  

If IoT MD $i$ chooses to offload its tasks to cloud computing OSP $j$, the input data transmitted to BS should be further delivered to the cloud computing center via the optical backbone network. To simplify our computation as \cite{guo2018collaborative}, the total wired transmission delay can be denoted as follows:
\begin{equation}
	\label{equ6}
	D_{i,j}^{tx,wired} = \left\{ {\begin{array}{*{20}{c}}
			{{a_j}\frac{{{z_i}}}{R} + t ,}&{1 \le j \le {N_c},}\\
			{0,}&{otherwise.}
	\end{array}} \right.	
\end{equation}

Here, $ a_j $ is the number of optical amplifiers from the BS to the cloud computation center owned by OSP $ j $, and $ R $ is the uplink data rate of the optical fiber network. $ t $ denotes the uplink propagation delay.
 
In the following, we will study the delay caused by task computation in the server of OSP $j$. When $ 1 \le j \le {N_c} $, an $ M/M/\infty $ queue is modeled for task computation in a cloud computation center owing to its rich computing resources. Given that there is only one edge server owned by each edge computing OSP, the service queue in an edge computing OSP is modeled as $ M/M/1 $. Therefore, we define the average computation delay of a task offloaded from IoT MD $ i $ to OSP $ j $ as follows:
\begin{equation}
	\label{equ7}
	D_{i,j}^{offload} = \left\{ {\begin{array}{*{20}{c}}
			{\frac{{{c_i}}}{{f_j^{OSP}}},}&{1 \le j \le {N_c},}\\
			{\frac{{{c_i}}}{{f_j^{OSP} - \sum\nolimits_{k \in {\cal M}} {{\alpha _{k,j}}{\lambda _k}{c_k}} }},}&{otherwise,}
	\end{array}} \right.
\end{equation}
where ${ f_j^{OSP}}$ is the service rate of a server owned by OSP $ j $. Matrix ${\bm{{A}}{_{ - i}}}$ represents the offloading strategies of all IoT MDs except $i$. From (\ref{equ7}), we can find that when a task of IoT MD $i$ is computed in an edge server, the delay is affected by offloading strategies of other IoT MDs.

Payment cost is also a significant cost in the process of task offloading. Assuming OSP $ j $ charges ${ p_j } $ for each unit CPU cycle, when IoT MD $ i $ is offloading tasks to OSP $ j $, the money $ i $ should pay to $ j $ for making use of computing resources can be denoted as follows:
\begin{equation}
	\label{equ8}
	P_{i,j} = {p_i}{c_i}{\lambda _i}.
\end{equation}

\subsection{Utility Function}
The QoE of an IoT MD can be affected by the computing delay, energy consumption and payment cost. Given the offloading strategy vector ${\bm{{\alpha}}{ _i}}$, from (\ref{equ1}), (\ref{equ2}), (\ref{equ4}), (\ref{equ5}), (\ref{equ6}), (\ref{equ7}) and (\ref{equ8}), we can get the expectation value of computing delay ${D_i}\left( {\bm{{\alpha}}{ _i},\bm{{A}}{_{ - i}}} \right)$, energy consumption ${E_i}\left( {\bm{{\alpha}}{ _i}} \right)$ and payment cost ${P_i}\left( {\bm{{\alpha}}{ _i}},\bm{p} \right)$ for IoT MD $i$ as follows:
\begin{equation}
	\label{equ9}
	\begin{array}{l}
		{D_i}\left(\bm{\alpha}_{i},\bm{A}_{-i}\right) = \left( {1 - \sum\nolimits_{j \in {\cal N}} {{\alpha _{i,j}}} } \right)D_i^{local}\\
		+\sum\nolimits_{j \in {\cal N}} {{\alpha _{i,j}}\left( {D_i^{tx,wireless} + D_{i,j}^{tx,wired} + D_{i,j}^{offload}} \right)}, 
	\end{array}
\end{equation}
\begin{equation}
	\label{equ10}
	{E_i}\left(\bm{\alpha}_{i}\right) = \left( {1 - \sum\nolimits_{j \in {\cal N}} {{\alpha _{i,j}}} } \right)E_i^{local} + \sum\nolimits_{j \in {\cal N}} {{\alpha _{i,j}}} E_i^{tx},
\end{equation}
\begin{equation}
	\label{equ11}
	{P_i}\left(\bm{\alpha}_{i},\bm{p}\right) = \sum\nolimits_{j \in {\cal N}} {\left( {{\alpha _{i,j}}{p_j}{c_i}{\lambda _i}} \right)},
\end{equation}
where $ \bm{p} $ is the pricing vector of all OSPs.

To jointly combine ${D_i}\left( {\bm{{\alpha}}{ _i},\bm{{A}}{_{ - i}}} \right)$, ${E_i}\left( {\bm{{\alpha}}{ _i}} \right)$ and ${P_i}\left( {\bm{{\alpha}}{ _i}},\bm{p} \right)$ as the quantified form of IoT MD QoE, the weighted method is introduced to solve this problem as \cite{liu2017multiobjective}. Different IoT MDs have distinct tolerances of computing delay, energy consumption and payment cost, thus we define three constants ${D_i^{MAX}}$, ${E_i^{MAX}}$ and ${P_i^{MAX}}$, which represent the maximum computing delay, energy consumption and payment cost the IoT MD $ i $ can accept. Furthermore, for IoT MD $ i $, there are three weight factors $\theta _i^ {\rm D}$, $\theta _i^{\rm E}$ and $\theta _i ^ {\rm P}$, where $\theta _i ^ {\rm D}  + \theta _i ^ {\rm E} + \theta _i ^ {\rm P} = 1$, reflecting the relative importance of computing delay, energy and payment for $ i $. In the distributed task offloading and pricing mechanism, these three weight factors corresponding to an IoT MD $ i $ are privatized by IoT MD $ i $, effectively protecting the privacy of IoT MDs. From (\ref{equ9}), (\ref{equ10}), and (\ref{equ11}), the disutility function of IoT MD $ i $ can be expressed as follows:
\begin{equation}
	\label{equ12}
	\begin{array}{l}
		U_i^{MD}\left( {\bm{{\alpha}} {_i},\bm{{A_{ - i}}},\bm{p}} \right)
		\\= \theta _i^{\rm D}\frac{{{D_i}\left( {\bm{{\alpha}}{ _i},{\bm{A}_{ - i}}} \right)}}{{{D^{MAX}}}} + \theta _i^{\rm E}\frac{{{E_i}\left( {\bm{{\alpha}}{ _i}} \right)}}{{{E^{MAX}}}} + \theta _i^{\rm P}\frac{{{P_i}\left( {\bm{{\alpha}}{ _i},\bm{p}} \right)}}{{{P^{MAX}}}}.
	\end{array}
\end{equation}

In this paper, we assume that all OSPs do not consider other expenditures apart from the revenue earned by charging for the computation offloading. Hence, the utility function of OSP $ j $ can be represented by the revenue function expressed as follows:
\begin{equation}
	\label{equ13}
	U_j^{OSP}\left( {{p_j},{\bm{p}_{ - j}},\bm{A}} \right) = {p_j}\sum\nolimits_{k \in {\cal M}} { {{\lambda _k}{c_k}{\alpha _{k,j}} }},
\end{equation}
where $ \bm{p}_{ - j} $ are the prices announced by all OSPs except $ j $. $ \bm{A} $ is the offloading strategy profile of all MDs. The reason why $ U_j^{OSP} $ is related to $ \bm{p}_{ - j} $ is that the other OSPs' prices can affect the offloading strategies of IoT MDs, thereby indirectly affecting the revenue of OSP $ j $. 

\section{Game analysis} \label{sec-4}
In this section, first of all, the optimization problems of each IoT MD and each OSP are formulated and a multi-leader multi-follower two-tier Stackelberg game model is applied to study the interaction between IoT MDs and OSPs. Then, we define the competition among IoT MDs and the competition among OSPs as non-cooperative games and prove that both have Nash equilibrium solutions, thus proving the existence of Stackelberg equilibrium (SE).

\subsection{Problem Formulation}
To this end, based on previous analytic results on the disutility function of each IoT MD, the optimization problem of IoT MD $ i $,  $ \forall i \in {\cal M} $, can be formulated as $ \bm {P1} $, which is
\begin{equation}
	\label{equ14}
	\begin{aligned}
		\mathop {\min }\limits_{ {{\bm{\alpha} _i}} }&\   U_i^{MD}({\bm{\alpha} _i},{\bm{A}_{ - i}},\bm{p})\\
		s.t. \ & C{\rm{1:}}\ 0 \le \sum\nolimits_{j \in {\cal N}} {{\alpha _{i,j}} \le 1}, \\
		\ & C{\rm{2:}}\ 0 \le {\alpha _{i,j}} \le 1,\forall j \in {\cal N},\\
		\ & C{\rm{3:}}\ \left( {1 - \sum\nolimits_{j \in {\cal N}} {{\alpha _{i,j}}} } \right){\lambda _i}{c_i} < f_i^{MD},\\
		\ & C{\rm{4:}}\ {\lambda_{i}}{z_i}{\sum\nolimits_{j \in {\cal N}} {{\alpha _{i,j}}}} < r_i,\\
		\ & C{\rm{5:}}\ {\sum\nolimits_{k \in {\cal M}} \left({\alpha _{k,j}}{\lambda_k}{c_k}\right)} < f_j^{OSP}, \forall j \in {\cal N},\\
		\ & C{\rm{6:}}\ {D_i}{\left(\bm{\alpha}_i, \bm{A}_{-i}\right)} \le D_i^{MAX},\\
		\ & C{\rm{7:}}\ {E_i}{\left(\bm{\alpha}_i\right)} \le E_i^{MAX},\\
		\ & C{\rm{8:}}\ {P_i}{\left(\bm{\alpha}_i, \bm{p}\right)} \le P_i^{MAX}.
	\end{aligned}
\end{equation}

As for OSPs, according to (\ref{equ7}), the optimization problem of OSP $ j $, $\forall j \in {\cal N}$, can be formulated as $ \bm{P2} $, which is
\begin{equation}
	\label{equ15}
	\begin{array}{l}
		\mathop {\max }\limits_{ {{p_j}}} U_j^{OSP}\left( {{\bm{p}_j},{\bm{p}_{ - j}},\bm{A}} \right),\\
		s.t. \ p_j \ge p_j^{min},
	\end{array}
\end{equation}
where $ p_j^{min} $ is the minimum price OSP $ j $ can announce, reflecting OSP $ j $'s operating cost. 

From the view of an OSP, having no idea of the utility functions of all IoT MDs, each OSP can only dynamically adjust its price according to the amount of tasks offloaded to it. If the price is too high, IoT MDs may offload less tasks to it. Instead, if a low price is announced, a reduced utility may appear. As such, a feasible price that can maximize the revenue is pursued by each OSP.

\subsection{Multi-leader multi-follower two-tier Stackelberg game}
Problems $ \bm{P1} $ and $ \bm{P2} $ together form a multi-leader multi-follower two-tier Stackelberg game. The leaders of the game are all OSPs and all MDs are followers. Consequently, the player set of the game consists of sets ${\cal M}$ and ${\cal N}$. The objective of this game is to find the Stackelberg equilibrium (SE) solution(s) from which neither the leaders nor the followers have incentive to deviate. The formal definition of the SE is given as follows.
\begin{definition}
	(Stackelberg equilibrium, SE)
	Let $ \bm{\alpha}_i^{*} $ be a solution for $ \bm{P1} $ of IoT MD $ i $ and $ p_j^* $ be a solution for $ \bm{P2} $ of OSP $j$. Then, the pair $\left( {{\bm{A} ^ * },{\bm{p}^ * }} \right)$ is an SE for the proposed Stackelberg game if for any $ \bm{A}  \in \bm{{\cal A}}$ and $\bm{p} \in \bm{{\cal P}} $, where $ \bm{{\cal A}} \buildrel \Delta \over = {\left( {{\cal A}{_i}} \right)_{i \in {\cal M}}} $ and $ \bm{{\cal P}} \buildrel \Delta \over = {\left( {{\cal P}{_i}} \right)_{i \in {\cal M}}} $ are respectively the strategy set of MDs and OSPs, the following conditions are satisfied:
	\begin{equation}
		\label{equ16}
		U_i^{MD}(\bm{\alpha} _i^*,\bm{A} _{ - i}^*,{\bm{p}^*}) \le U_i^{MD}({\bm{\alpha} _i},\bm{A} _{ - i}^*,{\bm{p}^*})
	\end{equation}
	and
	\begin{equation}
		\label{equ17}
		U_j^{OSP}(p_j^*,\bm{p}_{ - j}^*,{\bm{A} ^*}) \ge U_i^{OSP}({p_j},\bm{p}_{ - j}^*,{\bm{A} ^*}).
	\end{equation}
\end{definition}

In our proposed distributed computation offloading and pricing mechanism, we consider both IoT MDs and OSPs to be selfish, which means that IoT MDs compete with each other for OSPs' computing resources and each OSP adjust its price to improve its revenue in a fully non-cooperated manner. Therefore, the Stackelberg game is composed of a follower non-cooperative game and a leader non-cooperative game. If both the non-cooperative game of followers and leaders are sequentially proven to achieve the Nash equilibrium, an SE solution surely exists in our proposed Stackelberg game.

\subsection{Follower Non-Cooperative Game Analysis}
As followers, given the prices of all OSPs, each IoT MD competes with each other to minimize its own disutility. Formally, we define the follower non-cooperative game as ${{\cal G}^{follower}}{\rm{ = }}\left( {{\cal M},{\left\{ {{{\cal A}_i}} \right\}_{i \in {\cal M}}},{{\left\{ {{U_i^{MD}}} \right\}}_{i \in {\cal M}}}} \right)$, where ${\cal M}$ is the set of players. Besides, the payoff profile of ${{\cal G}^{follower}}$ consists of the disutility functions of all IoT MDs as shown by ${{{\left\{ {{U_i^{MD}}} \right\}}_{i \in {\cal M}}}}$. The Nash equilibrium (NE) of $ {\cal G}^{follower} $ is defined as follows.

\begin{definition}
	\label{def2}
	(Nash equilibrium (NE) of $ {\cal G}^{follower} $) From (\ref{equ14}), the Nash equilibrium of game ${{\cal G}^{follower}}$ defined above is a feasible offloading strategy profile ${\bm{A}^ * } \buildrel \Delta \over = {\left( {\bm{\alpha} _i^ * } \right)_{i \in {\cal M}}}$, such that each ${\bm{\alpha} _i^ * }$ belongs to ${{{\cal A}_i}}$ and satifies the condition as follows:
	\begin{equation}
		\label{equ18}
		\bm{\alpha} _i^ *  \in \arg \mathop {\min }\limits_{{\bm{\alpha} _i}} {U_i^{MD}}\left( {{\bm{\alpha} _i},\bm{A} _{ - i}^ *} \right),
	\end{equation}
	which is subject to $ C1 $-$ C8 $.
\end{definition}

\begin{theorem}
	\label{thm1}
	For $\forall i \in {\cal M}$, the set ${{\cal A}_i}$ is closed and convex and the function ${U_i^{MD}}\left( {\bm{{\alpha}}{ _i},\bm{{A}}{_{ - i}}} \right)$ is continuously differentiable on $ \bm{{\cal A}} \buildrel \Delta \over = {\left( {{\cal A}{_i}} \right)_{i \in {\cal M}}} $ and convex in ${\bm{{\alpha}}{ _i}}$ for every fixed $\bm{{A}}{_{ - i}} \in {\bm{{\cal A}}_{ - i}}$.
	\begin{proof}
		\label{prf1}
		Accoding to \cite{boyd2004convex}, a function is convex if and only if its Hessian matrix is positive semidefinite. Thus, to prove the convexity of the objective function ${U_i^{MD}}\left( {\bm{{\alpha}}{ _i},\bm{{A}}{_{ - i}}} \right)$, it suffices to show that for every fixed $\bm{{A}}{_{ - i}} \in {\bm{{\cal A}}_{ - i}}$,
		\begin{equation}
			\label{equ19}
			\nabla _{\bm{{\alpha}}{ _i}}^2{U_i^{MD}} \succeq 0.
		\end{equation}
		From (\ref{equ12}), (\ref{equ19}) can be converted to:
		\begin{equation}
			\label{equ20}
			\theta _i ^ {\rm D}\frac{{\nabla _{\bm{{\alpha}}{ _i}}^2{D_i}}}{{{D_i^{MAX}}}} + \theta _i ^ {\rm E}\frac{{\nabla _{\bm{{\alpha}}{ _i}}^2{E_i}}}{{{E_i^{MAX}}}} + \theta _i ^ {\rm P}\frac{{\nabla _{\bm{{\alpha}}{ _i}}^2{P_i}}}{{{P_i^{MAX}}}} \succeq 0.
		\end{equation}
		
		The (x,y)-th elements of ${\nabla _{{\bm{\alpha} _i}}^2{D_i}}$, ${\nabla _{{\bm{\alpha} _i}}^2E}$ and ${\nabla _{{\bm{\alpha} _i}}^2P}$ are calculated as follows:
		\begin{equation}
			\label{equ21}
			\begin{aligned}
				{\left( {\nabla _{{\alpha _i}}^{\rm{2}}{D_i}} \right)_{xy}} 
				& = \frac{{{\partial ^2}{D_i}}}{{\partial {\alpha _{i,x}}{\alpha _{i,y}}}}\\
				& 
				 = 
				 \left\{ {\begin{array}{*{20}{c}}
				 		{\begin{array}{*{20}{c}}
				 				{{\Gamma _i} + {\Upsilon _i} + {\Psi _{ix}},}&{x = y,{N_c} + 1 \le x \le N}
				 		\end{array}}\\
				 		{\begin{array}{*{20}{c}}
				 				{{\Gamma _i} + {\Upsilon _i},}&{otherwise},
				 		\end{array}}
				 \end{array}} \right.
			\end{aligned}
		\end{equation}
		\begin{equation}
			\label{equ22}
			{\left( {\nabla _{{\alpha _i}}^{\rm{2}}{E_i}} \right)_{xy}} = \frac{{{\partial ^2}{E_i}}}{{\partial {\alpha _{i,x}}{\alpha _{i,y}}}} = \varepsilon _i^{local}{\Gamma _i} + \varepsilon _i^{tx}{\Upsilon _i},
		\end{equation}
		\begin{equation}
			\label{equ23}
			{\left( {\nabla _{{\alpha _i}}^2{P_i}} \right)_{xy}} = \frac{{{\partial ^2}{P_i}}}{{\partial {\alpha _{i,x}}\partial {\alpha _{i,y}}}} = 0,
		\end{equation}
		where,
		\begin{equation}
			\label{equ24}
			\begin{aligned}
				{\Gamma _i}{\rm{ = }} & \frac{{{\rm{2}}{\lambda _i}c_i^2}}{{{{\left[ {f_i^{MD} - \left( {1 - \sum\nolimits_{j \in {\cal N}} {{\alpha _{i,j}}} } \right){\lambda _i}{c_i}} \right]}^2}}}\\
				&  + \left( {1 - \sum\nolimits_{j \in {\cal N}} {{\alpha _{i,j}}} } \right)\frac{{2\lambda _i^2c_i^3}}{{{{\left[ {f_i^{MD} - \left( {1 - \sum\nolimits_{j \in {\cal N}} {{\alpha _{i,j}}} } \right){\lambda _i}{c_i}} \right]}^3}}},
			\end{aligned}
		\end{equation}
		\begin{equation}
			\label{equ25}
			\begin{aligned}
				{\Upsilon _i} = & \frac{{{\lambda _i}{\overline {S_i^2} }}}{{{{\left( {1 - \frac{{{\lambda _i}{z_i}\sum\nolimits_{j \in {\cal N}} {{\alpha _{i,j}}} }}{{{r_i}}}} \right)}^2}}}\\
				&  + \sum\nolimits_{j \in {\cal N}} {{\alpha _{i,j}}\frac{{{z_i}{{\overline {S_i^2} }}\lambda _i^2}}{{{r_i}{{\left( {1 - \frac{{{\lambda _i}{z_i}\sum\nolimits_{j \in {\cal N}} {{\alpha _{i,j}}} }}{{{r_i}}}} \right)}^3}}}},
			\end{aligned}
		\end{equation}
		\begin{equation}
			\label{equ26}
			\begin{aligned}
				{\Psi _{ix}} = & \frac{{2{\lambda _i}c_i^2}}{{{{\left( {f_x^{OSP} - \sum\nolimits_{k \in {\cal M}} {{\alpha _{k,x}}} {\lambda _k}{c_k}} \right)}^2}}}\\
				& + {\alpha _{i,x}}\frac{{2\lambda _i^2c_i^3}}{{{{\left( {f_x^{OSP} - \sum\nolimits_{k \in {\cal M}} {{\alpha _{k,x}}{\lambda _k}{c_k}} } \right)}^3}}}.	
			\end{aligned}	
		\end{equation}
		
		From constraints $ C1 $-$ C5 $, it is not hard to deduce that $ {\Gamma _i} > 0 $, $ {\Upsilon _i} > 0 $ and $ {\Psi _{ix}} > 0 $, for $ \forall i \in {\cal M} $ and $ \forall x \in {\cal N} $.
		
		After calculation, we find that $ \nabla _{\bm{\alpha} _i}^2{D_i} $, $ \nabla _{\bm{\alpha} _i}^2{E_i} $ and $ \nabla _{\bm{\alpha} _i}^2{P_i} $ are all positive (semi)definite matrices. Therefore, we can easily conclude that (\ref{equ17}) is valid. 
		
		Besides, we can easily find that the set ${{\cal A}_i}$ is closed for $\forall i \in {\cal M}$. To prove the convexity of the set ${{\cal A}_i}$, we need to confirm that the nonlinear constraints of objective function are convex. The convexity of nonlinear constraints can also be proved in the same manner as above.  
	\end{proof}
\end{theorem}

\begin{theorem}
	\label{thm2}
	The game $ {\cal G}^{follower} $ is equivalent to the variational inequality problem $ VI\left( \bm{{\cal A}},\bm{F} \right) $, where
	\begin{equation}
		\label{equ27}
		{\bm{F}(\bm{{\alpha}}{ _i},\bm{{A}}{_{ - i}})} = {\left( {{F_i}(\bm{{\alpha}}{ _i},\bm{{A}}{_{ - i}})} \right)_{i \in {\cal M}}},
	\end{equation}
	with
	\begin{equation}
		\label{equ28}
		{F_i}(\bm{{\alpha}}{ _i},\bm{{A}}{_{ - i}}) = {\nabla _{\bm{{\alpha}}{ _i}}}{U_i^{MD}}\left( {\bm{{\alpha}}{ _i},\bm{{A}}{_{ - i}}} \right).
	\end{equation}
	\begin{proof}
		\label{prf2}
		According to Proposition 2 in \cite{scutari2014real}, the game $ {\cal G}^{followr} $ can be equivalent to the variational inequality problem $VI\left( \bm{{\cal A}},\bm{F} \right)$, supposing that for each player IoT MD $i$, the following conditions hold:
		\begin{enumerate}
			\item \emph{the (nonempty) strategy set ${\cal A}_i$ is closed and convex;}
			\item \emph{the payoff function ${U_i^{MD}}\left( {\bm{{\alpha}}{ _i},\bm{{A}}{_{ - i}}} \right)$ is convex and continuously differentiable in ${\bm{{\alpha}}{ _i}}$ for every fixed $\bm{{A}}{_{ - i}}$ .}
		\end{enumerate}
		From Theorem 2, we can find that the two conditions are both satisfied by our formulated follower game ${\cal G}^{follower}$. Thus, the result follows.
	\end{proof}
\end{theorem}

The reason why we carry out the well-developed VI theory to analyze follower non-cooperative game is that the properties of $ VI\left( {\bm{{\cal A}},\bm{F}} \right) $ can reflect the existence/uniqueness of $ {\cal G}^{follower} $'s Nash equilibrium. Next, Theorem \ref{thm3} is proposed to prove that the $ {\cal G}^{follower} $ has such Nash equilibrium as Definition \ref{def2}.

\begin{theorem}
	\label{thm3}
	The mapping $\bm{F}$ is a monotone function on $\bm{{\cal A}}$, and the game $ {\cal G}^{follower} $ has a convex Nash equilibrium solution set.
	\begin{proof}
		\label{prf3}
		To begin with, we will classify the mapping $\bm{F}$ as a monotone function on $\bm{{\cal A}}$. According to Definition 40 in \cite{scutari2014real}, $\bm{F}$ is monotone on $\bm{{\cal A}}$ if for all pairs $\bm{x}$ and $\bm{y}$ in $\bm{{\cal A}}$,
		\begin{equation}
			\label{equ29}
			{\left( {\bm{x} - \bm{y}} \right)^T}\left( {F(\bm{x}) - F(\bm{y})} \right) \ge 0
		\end{equation}
		that is
		\begin{equation}
			\label{equ30}
			\sum\limits_{i = 1}^M {{{\left( {\bm{{x}}{_i} - \bm{{y}}{_i}} \right)}^T}\left( {{F_i}(\bm{x}) - {F_i}(\bm{y})} \right)}  \ge 0.
		\end{equation}
		
		According to \cite{liu2018game}, if for $\forall i \in {\cal M}$, ${F_i}$ meets the condition as follows:
		\begin{equation}
			\label{equ31}
			\sum\limits_{i = 1}^M {{{\left( {\bm{{x}}{_i} - \bm{{y}}{_i}} \right)}^T}\left( {{F_i}(\bm{x}) - {F_i}(\bm{y})} \right)}  \ge 0,
		\end{equation}
		(\ref{equ28}) holds.
		From Theorem 2, we can see that the Jacobian matrix ${J_{{\bm{\alpha} _i}}}{F_i} = \nabla _{{\bm{\alpha} _i}}^2U_i^{MD}$ is positive definite and condition (\ref{equ29}) holds. Therefore, we can conclude that  $\bm{F}$ is monotone on $\bm{{\cal A}}$. Next, as described in Theorem 3 of \cite{scutari2014real}, based on the monntonity of $\bm{F}$ on $\bm{{\cal A}}$, the game ${\cal G} ^ {follower}$ has a convex Nash equilibrium solution set.
	\end{proof}
\end{theorem}
\subsection{Leader Non-Cooperative Game Analysis}     
As leaders, having no idea of the utility functions of IoT MDs, OSPs can only dynamically adjust their prices according to the amount of tasks offloaded to them. The relationship among OSPs is competitive with no enforced rules and an OSP has no idea of other OSPs' pricing strategies. Therefore a leader non-cooperative game for OSPs is defined as $ {{\cal G}^{leader}}{\rm{ = }}\left( {{\cal N},{\left\{ {{{\cal P}_j}} \right\}_{j \in {\cal N}}},{{\left\{ {{U_j^{OSP}}} \right\}}_{j \in {\cal N}}}} \right) $, where $ \cal N $ denotes the set of players and the strategy set of OSP $ j $ is all prices above cost, i.e.,
$ {\cal P}_j = \left\lbrace p_j | p_j \ge p_j^{min}\right\rbrace $. The Nash equilibrium (NE) of $ {\cal G} ^ {leader} $ is defined as follows.
\begin{definition}
	\label{def3}
	(Nash equilibrium (NE) of $ {\cal G} ^ {leader} $) From (\ref{equ15}), the Nash equilibrium of game ${{\cal G}^{leader}}$ defined above is a feasible pricing strategy profile ${\bm{p}^ * } \buildrel \Delta \over = {\left( {p _j^ * } \right)_{j \in {\cal N}}}$, such that each $ p _j ^ * $ belongs to $ {\cal P}_j $ and satifies the conditions as follows:
	\begin{equation}
		\label{equ32}
		\begin{aligned}
			p _j ^ *  \in & \arg \mathop {\max }\limits_{ p_j} {U_j^{OSP}}\left( { p_j,\bm{p} _{ - j}^ *} \right),\\
			s.t. \ & p_j \ge p_j^{min}.
		\end{aligned}		
	\end{equation}
\end{definition}
Here, we consider that the offloading strategies of all IoT MDs are only determined by OSPs' prices. Therefore, the utility of the focused OSP will be limited only by its own price if the prices of other OSPs are given. Based on this, the following theorem about the existence of NE of $ {\cal G} ^ {leader} $ is analyzed.

\begin{theorem}
	\label{thm4}
	Toward the leader non-cooperative game $ {\cal G} ^ {leader} $, there always exists at least one NE, if for $\forall i \in \cal M$, the following condition holds:
	\begin{equation}
		\label{equ33}
		2{\Pi _i}\left[ {\frac{{c_i^3}}{{{{\left( {f_i^{MD} - {\lambda _i}{c_i}} \right)}^3}}} + \frac{{{\lambda _i}c_i^4}}{{{{\left( {f_i^{MD} - {\lambda _i}{c_i}} \right)}^5}}}} \right] \le {\Theta _i}\frac{{\overline {{s_i}} {z_i}}}{{{r_i}}},
	\end{equation}
	where
	\begin{equation}
		\label{equ34}
		{\Pi _i} = \frac{{P_i^{MAX}}}{{\theta _i^{\rm P}{\lambda _i}{c_i}}}\left( {\frac{{\theta _i^{\rm D}}}{{D_i^{MAX}}} + \frac{{\theta _i^{\rm E}\varepsilon _i^{local}}}{{E_i^{MAX}}}} \right),
	\end{equation}
	\begin{equation}
		\label{equ35}
		{\Theta _i} = \frac{{P_i^{MAX}}}{{\theta _i^{\rm P}{\lambda _i}{c_i}}}\left( {\frac{{\theta _i^{\rm D}}}{{D_i^{MAX}}} + \frac{{\theta _i^{\rm E}\varepsilon _i^{tx}}}{{E_i^{MAX}}}} \right).
	\end{equation}
\end{theorem}
\begin{proof}
	From the analysis of follower non-cooperative game $ {\cal G} ^ {follower} $, for a focused OSP $ j $, if the prices of other OSPs are given, we can get the relationship between $ p_j $ and $ {\alpha}_{i,j} $ from the first-order partial derivative of IoT MD $ i $'s disutility functions, i.e.,
	\begin{equation}
		\label{equ36}
		p_j = {RE}_i\left({{\alpha}_{i,j}} \right) = {p_j} - \frac{{P_i^{MAX}}}{{\theta _i^P{\lambda _i}{c_i}}} \cdot \frac{{\partial U_i^{MD}}}{{\partial {\alpha _{i,j}}}},\forall i \in {\cal M}.	
	\end{equation}
	
	Therefore, for each OSP $ j $, adjusting a proper price to maximize its utility is equivalent to sovling the problem as follows:
	\begin{equation}
		\begin{aligned}
			\label{equ37}
			\mathop {\max }\limits_{{p_j}} & {p_j}\sum\nolimits_{i \in {\cal M}} {{RE_i^{ - 1}({p_j})}{\lambda _i}{c_i}},\\
			s.t. \ & {p_j} \ge {p _j ^{min}},
		\end{aligned}
	\end{equation}
	
	To prove the existence of NE, we must prove that (\ref{equ37}) is concave, in other words, we should prove that the second derivative of the objective is always not greater than 0, i.e.,
	\begin{equation}
		\label{equ38}
		\begin{aligned}
			& \sum\nolimits_{i \in {\cal M}} {{\lambda _i}{c_i}\left[ {2{{\left( {RE_i^{ - 1}} \right)}^\prime } + {p_j}{{\left( {RE_i^{ - 1}} \right)}^{\prime \prime }}} \right]} \\
			= & \sum\nolimits_{i \in {\cal M}} {{\lambda _i}{c_i}\left[ {2\frac{1}{{{{\left( {R{E_i}} \right)}^\prime }}} + {p_j}\frac{{ - {{\left( {R{E_i}} \right)}^{\prime \prime }}}}{{{{\left( {{{\left( {R{E_i}} \right)}^\prime }} \right)}^3}}}} \right]}  \le 0,
		\end{aligned}
	\end{equation}
	where $ {{{\left( {R{E_i}} \right)}^\prime }} $ and $ {{{\left( {RE_i^{ - 1}} \right)}^\prime }} $ are the first derivative functions of $ {RE}_i\left({{\alpha}_{i,j}} \right) $ and its inverse function respectively. $ {{{\left( {R{E_i}} \right)}^{\prime \prime }}} $ and $ {{{\left( {RE_i^{ - 1}} \right)}^{\prime \prime }}} $ are the second derivative functions of $ {RE}_i\left({{\alpha}_{i,j}} \right) $ and its inverse function.
	
	After calculation, we can easily find that $ {{{\left( {R{E_i}} \right)}^\prime }} $ is always less than 0. In order to ensure that for $ \forall p_j \ge p_j^{min} $, inequality (\ref{equ38}) always holds, $ {{{\left( {R{E_i}} \right)}^{\prime \prime }}} \le 0 $ should always hold. By analyzing the analytical formula of $ {{{\left( {R{E_i}} \right)}^{\prime \prime }}} $, we find that when $ \bm{\alpha}_i $ takes $ \bm{0} $, the value of $ {{{\left( {R{E_i}} \right)}^{\prime \prime }}} $ is largest. According to $ max\left\lbrace {{{\left( {R{E_i}} \right)}^{\prime \prime }}} \right\rbrace \le 0 $, the condition (\ref{equ33}) holds. Therefore, the proof of Theorem \ref{thm4} is proved.      	
\end{proof}
 
\section{Algorithm Design} \label{sec-5}
In this section, based on the game theory analysis in Sec. \ref{sec-4}, we propose a fully distributed task offloading and pricing mechanism in MEC enabled edge-cloud systems as Fig. \ref{fig:distributed_mechanism}. In the leader non-cooperative game, the OSPs must wait for the follower non-cooperative game to achieve Nash equilibrium after announcing their prices. The collection and distribution of OSPs' prices and IoT MDs' offloading strategies are all done through the BS. Some public knowledge includes the average task arrival rate and average number of CPU cycles required for each task of each IoT MD is learnt by BS as soon as the IoT MD moves into the coverage of the BS. 

This section is composed of two subsections. In the first subsection, we will deal with the first condition where the prices charged by all OSPs are confirmed, and a distributed iterative algorithm is proposed to obtain equilibrious offloading strategies for all IoT MDs. In the second subsection, we put forward an algorithm to determine OSPs' pricing strategies iteratively and to achieve approximate Stackelberg equilibrium pursuantly. Both of these iterative algorithms ensure little communication overhead thanks to the negligible size of the communication message.
   
\begin{figure}[!t]
	\centering
	\includegraphics[width=3.5in]{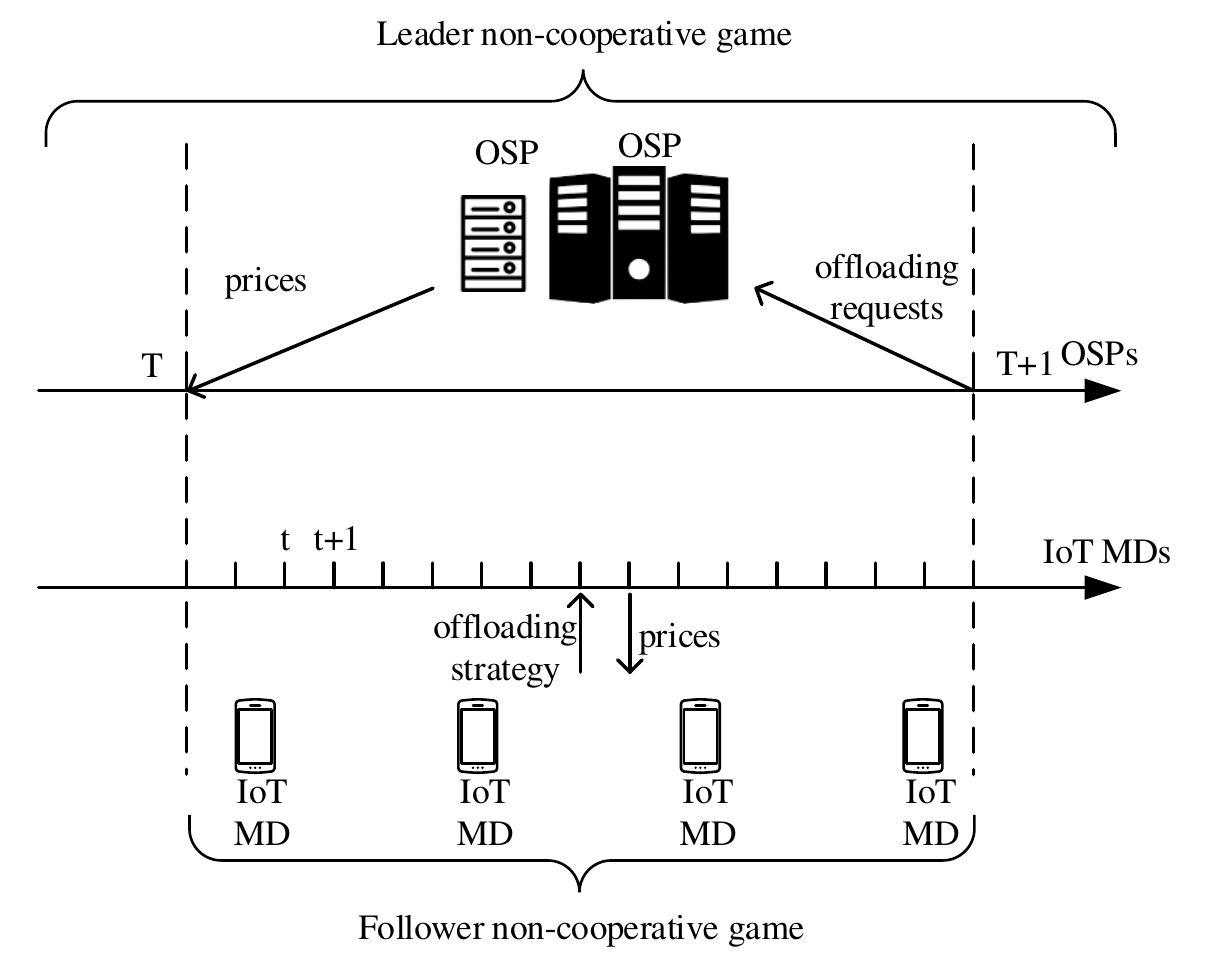}
	\DeclareGraphicsExtensions.
	\caption{Distributed task offloading and pricing mechanism}
	\label{fig:distributed_mechanism}
\end{figure}

\subsection{IPOA}
As described in \cite{scutari2014real}, there are a host of solution methods available in the literature to solve monotone real VIs and thus monotone Nash equilibrium problems, but these algorithms are centralized and unsuitable in the absence of authority. With regards to this matter, we will introduce an iterative proximal offloading algorithm (IPOA) to solve the follower non-cooperative game ${{\cal} G^{follower}}$ that has multiple Nash equilibrium solutions. The iterative algorithm allows all IoT MDs to update their strategies simultaneously in a decentralized manner.

Before beginning introducing the IPOA in detail, we consider a perturbation of game ${{\cal G}^{follower}}$ defined as ${{\cal G}_{\tau ,\bm{\beta} }^{follower}} = \left( {{\cal M},{\cal A},{{\left[ {{U_i^{MD}} + \left( {\tau /2} \right){{\left\| { \bm{\alpha}_i  - {\bm{\beta} _i}} \right\|}^2}} \right]}_{i \in {\cal M}}}} \right)$, where $\tau $ is a positive parameter and $\bm{\beta}  = {\left( {{\bm{\beta} _i}} \right)_{i \in {\cal M}}}$ with each $ {\bm{\beta} _i} \in {{R}^{N}}$. The Nash equilibrium of ${{\cal G}_{\tau ,\bm{\beta} }^{follower}}$ is each IoT MD $ i $ obtain its optimal offloading strategy $\bm{\alpha} _i^ *  \in {{\cal A}_i}$ to solve the following convex optimization problem:
\begin{equation}
	\label{equ39}
	\mathop {\min }\limits_{{\bm{\alpha} _i}} {U_i^{MD}}\left( {{\bm{\alpha} _i},{\bm{A}_{ - i}}} \right) + \frac{\tau }{2}{\left\| {{\bm{\alpha} _i} - {\bm{\beta} _i}} \right\|^2},
\end{equation}
is subject to $ C1 $-$ C8 $. When $ \tau $ is large enough, ${\bm{A}^ * }$ is a Nash equilirium solution of ${\cal G}$ if and only if ${\bm{A}^ * }$ is a Nash equilirium solution of ${{\cal G}_{\tau ,{\bm{A}^ * }}}$. The way to obtain a Nash equilirium solution of ${{\cal G}_{\tau ,{\bm{A}^ * }}}$ is described thoroughly in Algorithm \ref{alg1}.

At the beginnig, given $ \bm{p} = {\left( {p_j} \right)_{j \in {\cal N}}}$, we initialize the offloading strategies of all IoT MDs to feasible values ${\bm{A}^{\left( 0 \right)}}$, and set the number of iterations to 0. Moreover, the BS sends the service capabilities collected from all OSPs to all IoT MDs. Within iteration $ l $, each IoT MD $ i $ deals with its convex optimization problem (\ref{equ39}) with several linear and nonlinear constraints using \emph{Interior Point Method} and sends its own offloading strategy to BS. After the last IoT MD sending its newly updated strategy to the BS, the BS broadcasts ${\bm{A}^{\left( {l} \right)}}$ to all IoT MDs and informs them to update their centroid (${\bm{\beta} _i}$). A new iteration begins unless the offloading strategies of all IoT MDs are fixed. In this paper, we assume that the offloading strategies of all IoT MDs are unchanged if $\left\| {{\bm{A}^{\left( l \right)}} - {\bm{A}^{\left( {l - 1} \right)}}} \right\| \le \sigma$ where ${\bm{A}^{\left( l \right)}} = {\left( {\bm{\alpha} _i^{(l)}} \right)_{i \in {\cal M}}}$ and $\sigma$ is a relativly small constant. As we can see, this algorithm effectively reduces the number of communications between OSPs and IoT MDs and it also guarantees the performance under the premise of ensuring IoT MDs' privacy. Given a $ \tau $ large enough, the game $ {\cal G}_{\tau, \beta}^{follower} $ can converge in finite iteration times. Due to space limitations, the detailed proof of the convergence of IPOA can be found in Theorem 16 in \cite{scutari2014real}. We denote $ I_{IPOA} $ as the iteration times of the outer for loop and the computational complexity of \emph{Interior Point Method} used by each IoT MD in each iteration can be considered as $ O\left(I_{IPM}\right) $ \cite{gondzio2012interior}, where $ I_{IPM} $ is a finite integer. Therefore, the total running time of IPOA can be calculated as $ O(I_{IPOA} \cdot I_{IPM} \cdot M) $. 
\begin{algorithm}
	\renewcommand{\algorithmicrequire}{\textbf{Input:}}
	\renewcommand{\algorithmicensure}{\textbf{Output:}}
	\caption{Iterative proximal offloading algorithm (IPOA)}
	\label{alg1}
	\begin{algorithmic}[1]
		\REQUIRE $\bm{p}$, $\sigma$, ${\cal M}$, $ \left\lbrace f_j^{OSP}\right\rbrace _{j \in {\cal N}}$
		\ENSURE ${\bm{A}^ * }$
		\STATE Initialization: choose a feasible offloading strategy profile ${\bm{A}^{\left( 0 \right)}}$. Set ${\bm{\beta} _i} \leftarrow \bm{\alpha}_i^{(0)},\forall i \in {\cal M}$ and $l \leftarrow 0$. The BS sends $\bm{p}$ and $ \left\lbrace f_j^{OSP}\right\rbrace _{j \in {\cal N}}$ to all IoT MDs.
		\REPEAT
		\FOR{each MD ${i \in {\cal M}}$}
		\STATE Receives $\bm{A}_{ - i}^{(l)}$ from the BS, and uses \emph{Interior Point Method} to obtain $\bm{\alpha} _i^{(l + 1)}$ as follows:
		\STATE $\bm{\alpha} _i^{(l + 1)} \in$
		\begin{equation*}
			\arg \min \left\{ {{U_i^{MD}}\left( {{\bm{\alpha} _i},\bm{A}_{ - i}^{(l)}} \right) + \frac{\tau }{2}{{\left\| {{\bm{\alpha} _i} - {\bm{\beta} _i}} \right\|}^2}} \right\}
		\end{equation*}
		\STATE Sends the newly updated offloading strategy ${\bm{\alpha}} _i^{(l + 1)}$ to the BS.
		\ENDFOR
		\IF{the optimal strategies of all IoT MDs in the $ l $-th iteration have been achieved}
		\STATE Each IoT MD updates its centrods, i.e., ${\bm{\beta} _i} \leftarrow \bm{\alpha} _i^{(l + 1)},\forall i \in {\cal M}$
		\ENDIF
		\STATE Set $l \leftarrow l + 1$
		\UNTIL{$\left\| {{\bm{A}^{\left( {l} \right)}} - {\bm{A}^{\left( l-1 \right)}}} \right\| \le \sigma $}
	\end{algorithmic}
\end{algorithm}

\subsection{ISPA}
Similar as \cite{wang2017multi} and \cite{8478374}, backward induction is used to solve the Stackelberg game and we propose an iterative Stackelberg game pricing algorithm (ISPA). The details of ISPA is described below.

We first set a feasible value $ p_j^{\left( 0 \right)} \ge p_j^{min}$ for the initial price of each OSP $j$, $j \in {\cal N}$ and obtain the corresponding offloading strategies of all IoT MDs through IPOA, then the initial utilities of all OSPs are confirmed. In each iteration $ k $, each OSP uses the marginal utility function to adjust its price based on IoT MDs offloading strategies, i.e.,
\begin{equation}
	\label{equ40}
	\begin{aligned}
		& p_j^{(k + 1)} \\ 
		& = \max \left( {p_j^{(k)} + {\Delta _j}\left( {\frac{{\partial U_j^{OSP}\left( {p_j^{(k)},\bm{p}_{ - j}^{(k)},{\bm{A}^{(k)}}} \right)}}{{\partial p_j^{(k)}}}} \right) ,p_j^{min}} \right),
	\end{aligned}
\end{equation}
where ${\Delta _j} > 0$ represents the positive iteration value of OSP $j$ and the partial derivative of ${U_j^{OSP}}$ to $p_j$ can be calculated by a relatively small value $ \eta $ as:
\begin{equation}
	\label{equ41}
	\begin{array}{l}
		\frac{{\partial U_j^{OSP}(p_j^{(k)},p_{ - j}^{(k)},{A^{(k)}})}}{{\partial p_j^{(k)}}}\\
		\approx \frac{{U_j^{OSP}\left( {...,p_j^{(k)} + \eta ,...} \right) - U_j^{OSP}(...,p_j^{(k)} - \eta ,...)}}{{2\eta }}
	\end{array}
\end{equation}
Here, we can find that as OSP $ j $ adjusts its price to gain the partial derivatinve of $ U_j^{OSP} $ to $ p_j $, IoT MDs should re-determine their offloading strategies to the small changes in price $ p_j $ via IPOA. In each iteration $ k $, each OSP gets its updated price, thereby a new collection of OSPs' utilities ${\bm{U}^{OSP}}\left( k \right) = {\left( {U_j^{OSP}(k)} \right)_{j \in {\cal N}}}$ is obtained. As the iteration increases, the prices gradually approximate the optimal prices. Taking into account the dynamics of the mobile network, in order to quickly obtain approximate equilibrium prices of all OSPs, we set $ I_{ISPA} $ as the maximum iteration times of outer for loop. Given the computational complexity of IPOA is $ O\left(I_{IPOA} \cdot I_{IPM} \cdot M\right) $, we can get that the total running time of ISPA is $ O\left(I_{ISPA} \cdot I_{IPOA} \cdot I_{IPM} \cdot M \cdot N\right) $. The pseudocode of the proposed algorithm is presented in Algorithm \ref{alg2}.

\begin{algorithm}
	\renewcommand{\algorithmicrequire}{\textbf{Input:}}
	\renewcommand{\algorithmicensure}{\textbf{Output:}}
	\caption{Iterative Stackelberg game pricing algorithm (ISPA)}
	\label{alg2}
	\begin{algorithmic}[1]
		\REQUIRE $\rho$, ${\cal N}$, ${\cal M}$
		\ENSURE ${\bm{p}^ *}$, ${\bm{A}^ * }$
		\STATE Initialization: set a feasible value for $ \bm{p}^{(0)} $. Get the corresponding offloading strategies of all IoT MDs ${\bm{A}^{\left( 0 \right)}}$ and calculate all OSPs' initial utilities $ U^{OSP}\left(0\right) $ consequently. Set $k \leftarrow 0$.
		\REPEAT
		\FOR{each OSP ${j \in {\cal N}}$}
		\STATE Adds an $ \eta $ to its price $ p_j^{\left(k\right)} $ and sends it to the BS, then corresponding offloading strategies of all IoT MDs are obtained via IPOA.
		\STATE Reduces an $ \eta $ to its price $ p_j^{\left(k\right)} $ and sends it to the BS, then corresponding offloading strategies of all IoT MDs are obtained via IPOA.
		\STATE Updates its price through (\ref{equ40}) and (\ref{equ41}) and sends the updated price to the BS.
		\ENDFOR
		\IF{all OSPs have updated their prices}
		\STATE The BS sends $ \bm{p}^{(k+1)} $ to IoT MDs and get the corresponding offloading strategies of all IoT MDs $ \bm{A}^{(k+1)} $. Calculate $ {\bm{U}^{OSP}}\left( k+1 \right) $ using  $ \bm{p}^{(k+1)} $ and $ \bm{A}^{(k+1)} $. 
		\ENDIF
		\STATE Set $k \leftarrow k + 1$.
		\UNTIL{$ k = I_{ISPA} $}
	\end{algorithmic}
\end{algorithm}

\section{Simulation Results} \label{sec-6}
In this section, we conduct extensive simulations to validate the performance of our proposed two algorithms for MEC enabled edge-cloud systems. For experimental purposes, we build up an IoT scenario consisting of a BS covering 
a circular area with a radius of 200m. In order to ensure the authenticity and credibility of the parameters, we refer to many papers and standards \cite{ebrahimzadeh2020cooperative,liu2017multiobjective,miettinen2010energy,yang2019mobile} to design our experimental simulation parameters. The detailed parameter settings are shown in Table. \ref{table2}. In Table. \ref{table2}, $ f_i^{MD} $ and $ f_j^{OSP} $ take the value uniformly in the corresponding value intervals and other parameters are designated fixed values. For $\forall i \in {\cal M}$, the weight factors, $\theta _i^{\rm D}$, $\theta _i^{\rm E}$ and $\theta _i^{\rm P}$ are randomly retrieved values from the range $\left[ {0,1} \right]$ so that they are individually different.
\begin{table}
	\caption{Simulation Parameters}
	\label{table2}
	\centering
	\setlength{\tabcolsep}{10pt}
	\begin{tabular}{cc|cc}
		\hline
		Parameter & Value & Parameter & Value \\
		\hline
		$ \lambda _i $ & $ \left[ 20, 29\right] $  task/min & $ \varepsilon _i^{tx} $ & $ \left[0.1, 1.0 \right] $ W \\
		$ c_i $ & 300 Mcycles & $ B $ & 100 MHz \\
		$ z_i $ & 500 Kb & $ w_0 $ & $ 10^{-8} W $ \\
		$ f_i^{MD} $ & $\left[300, 450\right]$ MHz & $ h_i $ & -50 dBm\\
		$ \varepsilon _i^{local} $ & 0.5 W & $ \sigma_i^{2} $ & 0\\
		$ D_i^{MAX} $& 1 sec &$ f_j^{OSP} $ & $ \left[1.44,2.9\right] $ GHz\\
		$ E_i^{MAX} $& 1 J &$ R $ & 10 Gbps\\
		$ P_i^{MAX} $& 0.1 \$ &$ t $ & 0\\
		\hline
	\end{tabular}
\end{table}

\subsection{Task offloading strategies of IoT MDs}
In this subsection, we will investigate the convergence and effectiveness of IPOA in an MEC enabled edge-cloud system that includes one cloud computing OSP and three edge computing OSPs. Given that the price charged by each cloud computing OSP is 0.2 \$/Gcycles and that the price of each edge computing OSP is 0.1 \$/Gcycles, we first assume that there are 50 IoT MDs within the coverage of the BS and the average task arrival rate of each IoT MD is 25 task/min. In addition, .As shown in Fig. \ref{fig:a1f1}, we select 5 IoT MDs (MD 5, 15, 25, 35 and 45), to show the disutility function values of these IoT MDs versus the number of iterations. From Fig. \ref{fig:a1f1}, we can observe that in the previous iterations, the change in the disutility function value of each IoT MD fluctuates drastically and reaches a relatively stable state in about 7 iterations. In other words, all IoT MDs achieve Nash equilibrium after a limited number of iterations given the prices of all OSPs. Moreover, considering that the time spent in each iteration is far less than the task computation time, the IPOA can converge to a Nash equilibrium very quickly and the high efficiency of IPOA is thus proved. 

\begin{figure}[!t]
	\centering
	\includegraphics[width=3.5in]{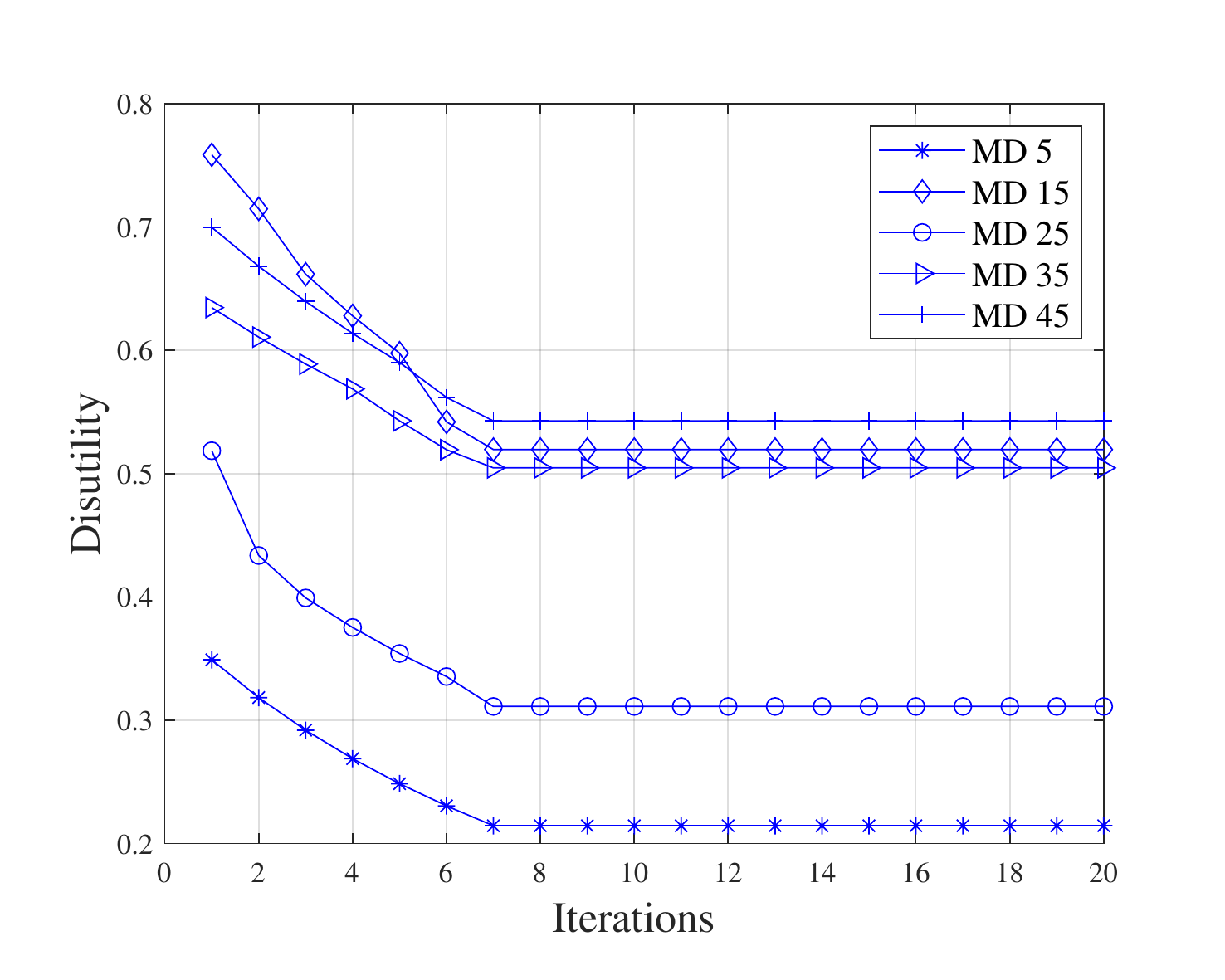}
	\DeclareGraphicsExtensions.
	\caption{The disutility versus the iterations of IPOA}
	\label{fig:a1f1}
\end{figure}

In Fig. \ref{fig:a1f2} and Fig. \ref{fig:a1f3}, we compare the average utility of all IoT MDs obtained through IPOA against the following baselines.
\begin{itemize}
	\item \emph{Local Computing}: Each IoT MD processes all of its tasks locally.
	\item \emph{Cloud Computing}: Each IoT MD offloads all of its tasks to cloud computing servers.
	\item \emph{Evenly Offloading}: Each IoT MD processes a task locally or each OSP with the same probability.
	\item \emph{Socially Optimal Offloading}: According to \cite{liu2017multiobjective}, given the prices of all OSPs, we can get the socially optimal offloading strategies of all IoT MDs via an IPM-based algorithm.    
\end{itemize}
In this paper, we define the price of anarchy (PoA), which reflects how far is the overall performance of an NE from the socially optimal offloading scheme, is the average disutility value obtained by IPOA divided by the socially optimal average disutility value, i.e.,     
\begin{equation}
	\label{poa}
	PoA = \frac{{{{\bar U}^{MD}}\left( {{\bm{A}^*}} \right) }}{{{{\bar U}^{MD}}\left( {{\bm{A}^{SO}}} \right)}},
\end{equation}
where $ {{{\bar U}^{MD}}\left( {{\bm{A}^*}} \right)} $ is IoT MDs' average disutility obtained by IPOA and $ {{{\bar U}^{MD}}({A^{SO}})} $ is the socially optimal average disutility of all IoT MDs. It is not difficult to see that the smaller the PoA, the closer the performance of IPOA is to the socially optimal performance.

As shown by Fig. \ref{fig:a1f2}, when we set the transmission power of each IoT MD to 400 mW, the IoT MDs' average disutility increases as the average task arrival rate of each IoT MD varies from 20 task/min to 29 task/min. Additionally, we can observe that the performance of IPOA is significantly superior to that of local computing and cloud computing. Although the average disutility obtained by IPOA is slightly less than evenly offloading, the system performance of IPOA evidently excels that of evenly offloading when considering the quantity base of IoT MDs. Based on (\ref{poa}), the value of PoA decreases from 43\% to 36\% as the average task arrival rate increases, which means the PoA of IPOA is bounded in the real scenes. 

Fig. \ref{fig:a1f3} demonstrates the impact of $\varepsilon_i^{tx} $ on the average disutility of all IoT MDs with the average task arrival rate of each IoT MD is set to a fixed value 25 task/min. Obviously, the average disutility caused by local computing is not altered by $
\varepsilon_i^{tx} $. On the contrary, the increase of $
\varepsilon_i^{tx} $ will bring more diutility for all IoT MDs when they offload their tasks through other offloading strategies. Fig. \ref{fig:a1f3} proves the high performance of IPOA compared with other baselines and once again verifies the NE obtained by IPOA does not bring immeasurable performance loss to the system.

\begin{figure}[!t]
	\centering
	\includegraphics[width=3.5in]{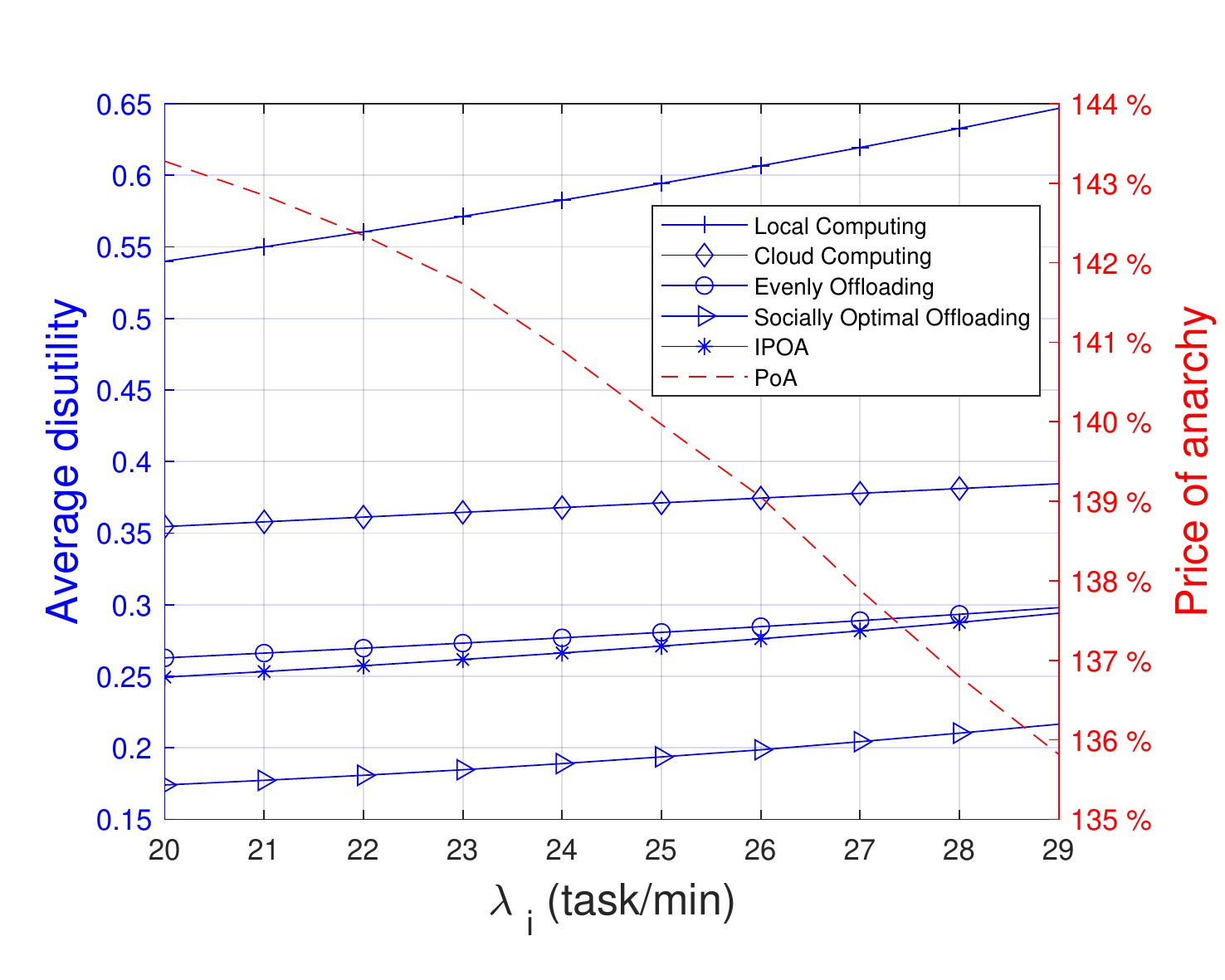}
	\DeclareGraphicsExtensions.
	\caption{The average disutility versus the average task arrival rate of IoT MDs}
	\label{fig:a1f2}
\end{figure}

\begin{figure}[!t]
	\centering
	\includegraphics[width=3.5in]{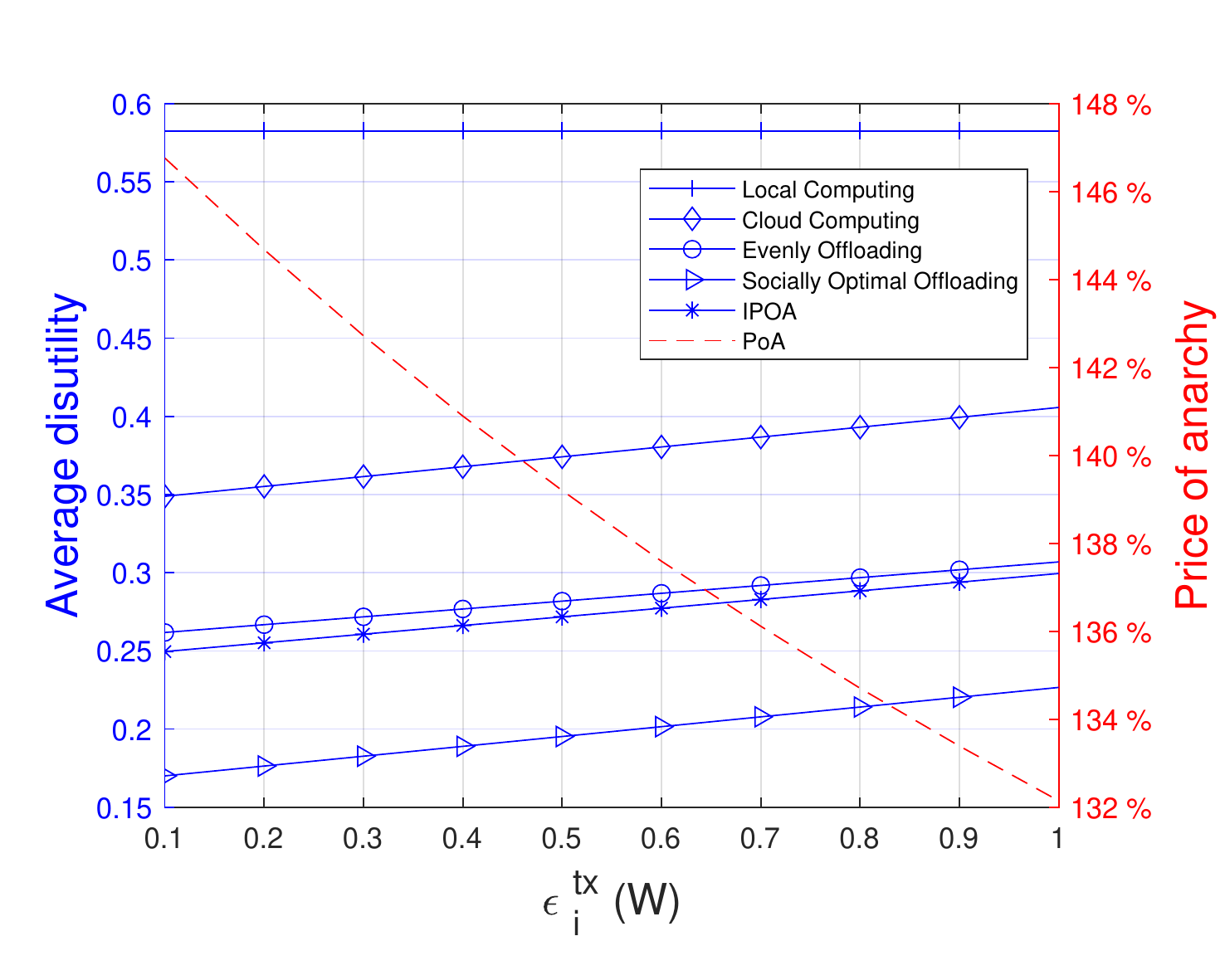}
	\DeclareGraphicsExtensions.
	\caption{The average disutility versus the transmission power of IoT MDs}
	\label{fig:a1f3}
\end{figure}

\subsection{Pricing strategies of OSPs}
In this subsection, we focus on the OSPs and study the performance of ISPA. Moreover, the influence of the IoT MD number on the performance is analyzed in the following experiment. In particular, we set the cost price of each OSP as its initial price at the beginning of ISPA.

In the case of setting $ \lambda _ i$ to 25 task/min and $ \varepsilon_i^{tx}$ to 400 mW, Fig. \ref{fig:a2f1} presents the changing prices of cloud computing OSPs and edge computing OSPs with the iterations. Specifically, the price variations corresponding to 5 different IoT MD numbers, 10, 30, 50, 70 and 90 respectively, are shown in different line types. The prices obviously show a rapidly increasing trend with the iteration increasing during the whole stage. The reason for the obvious fluctuation of the curve instead of monotonically increasing with the iterations is that we apply (\ref{equ41}) to approximate the derivative direction of each OSP utility function in each iteration. Furthermore, we can see that the prices of OSPs will increase significantly when the number of IoT MDs increases and the pricing of edge computing OSP is higher than that of cloud computing OSP as edge computing OSPs have the advantage of the cloud computing OSPs in the leader non-cooperative game.

Under the same simulation settings as in Fig. \ref{fig:a2f1}, Fig. \ref{fig:a2f2} illustrates the utility changes of the cloud computing OSPs and edge computing OSPs as the iterations progress with various IoT MD numbers. We can easily find that the utilities of all OSPs increase with each iteration. Besides, each edge computing OSP's utility is always higher than that of the cloud computing OSP. 

In Fig. \ref{fig:a2f3}, we compare ISPA with \textit{blind pricing} with 50 IoT MDs. Here, \textit{blind pricing} is a pricing scheme where each OSP increases its price without considering the impact of other OSPs' pricing. In our simulation, we assume that all blinding OSPs set their cost prices as their initial prices and linearly increase their prices to the average price of ISPA at the 50th iteration. As we can see, the average utility of OSPs in ISPA performs better than that of blinding OSPs after the 37th iteration and the performance gap between ISPA and blind pricing gradually widens. In the 50th iteration, ISPA can improve the average utility of OSPs by up to 7.5\% compared with the bliding pricing scheme at the 50th iteration and a greater improvement can be foreseen if the iteration continues.              

\begin{figure}[!t]
	\centering
	\includegraphics[width=3.5in]{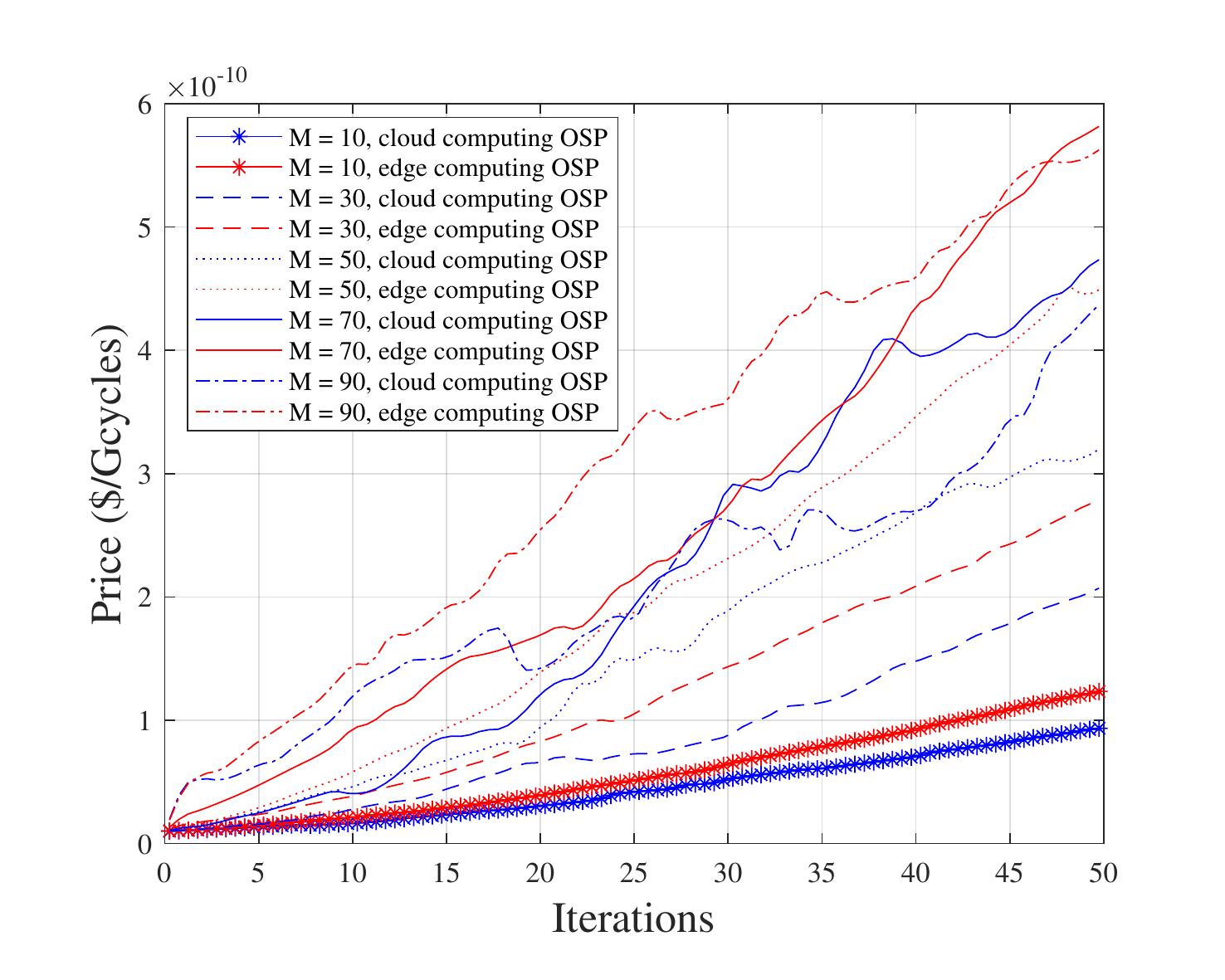}
	\DeclareGraphicsExtensions.
	\caption{Prices of OSPs versus the iterations of ISPA}
	\label{fig:a2f1}
\end{figure}

\begin{figure}[!t]
	\centering
	\includegraphics[width=3.5in]{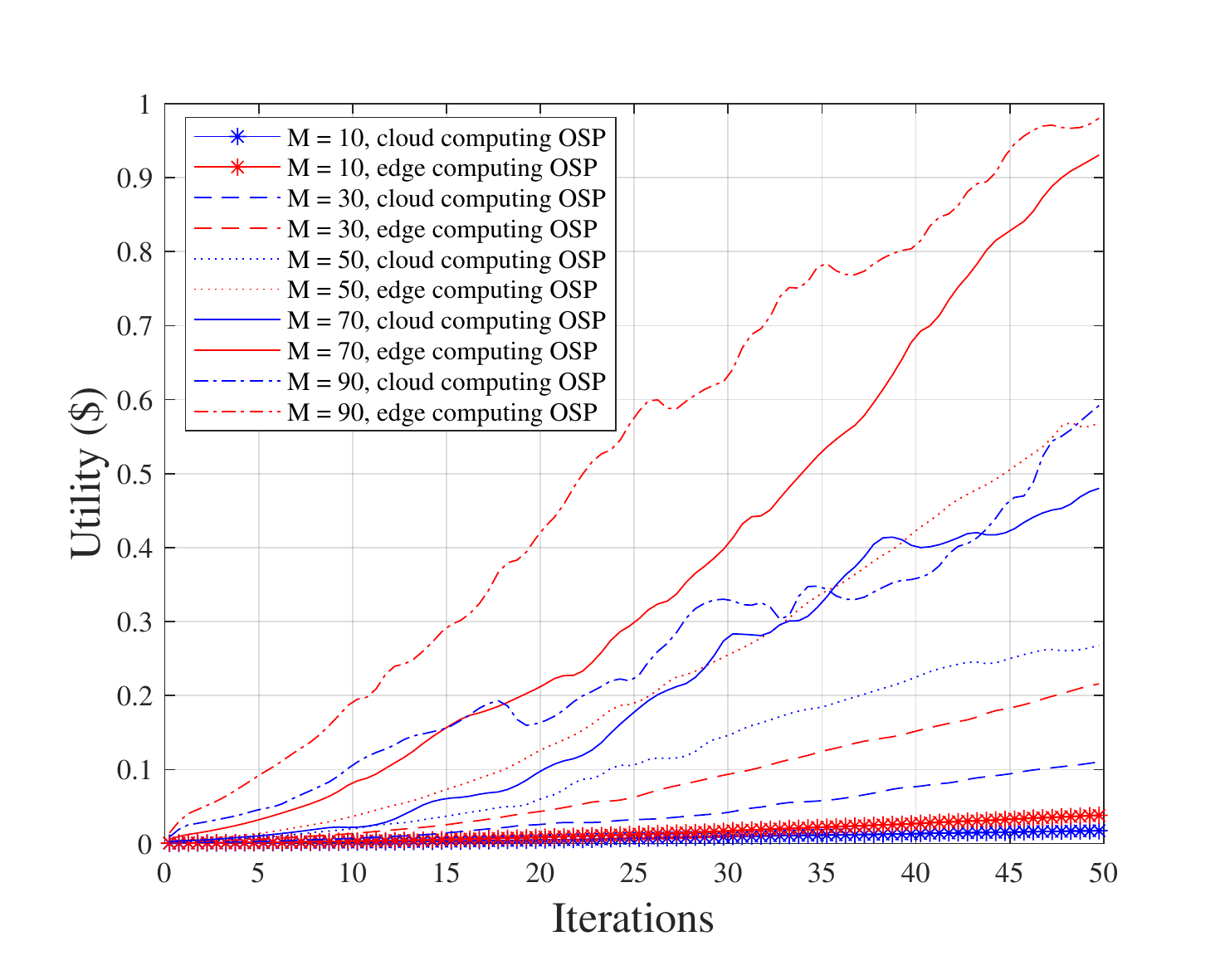}
	\DeclareGraphicsExtensions.
	\caption{Utilites of OSPs versus the iterations of ISPA}
	\label{fig:a2f2}
\end{figure}

\begin{figure}[!t]
	\centering
	\begin{subfigure}[b]{0.4\linewidth}
		\includegraphics[width=\linewidth]{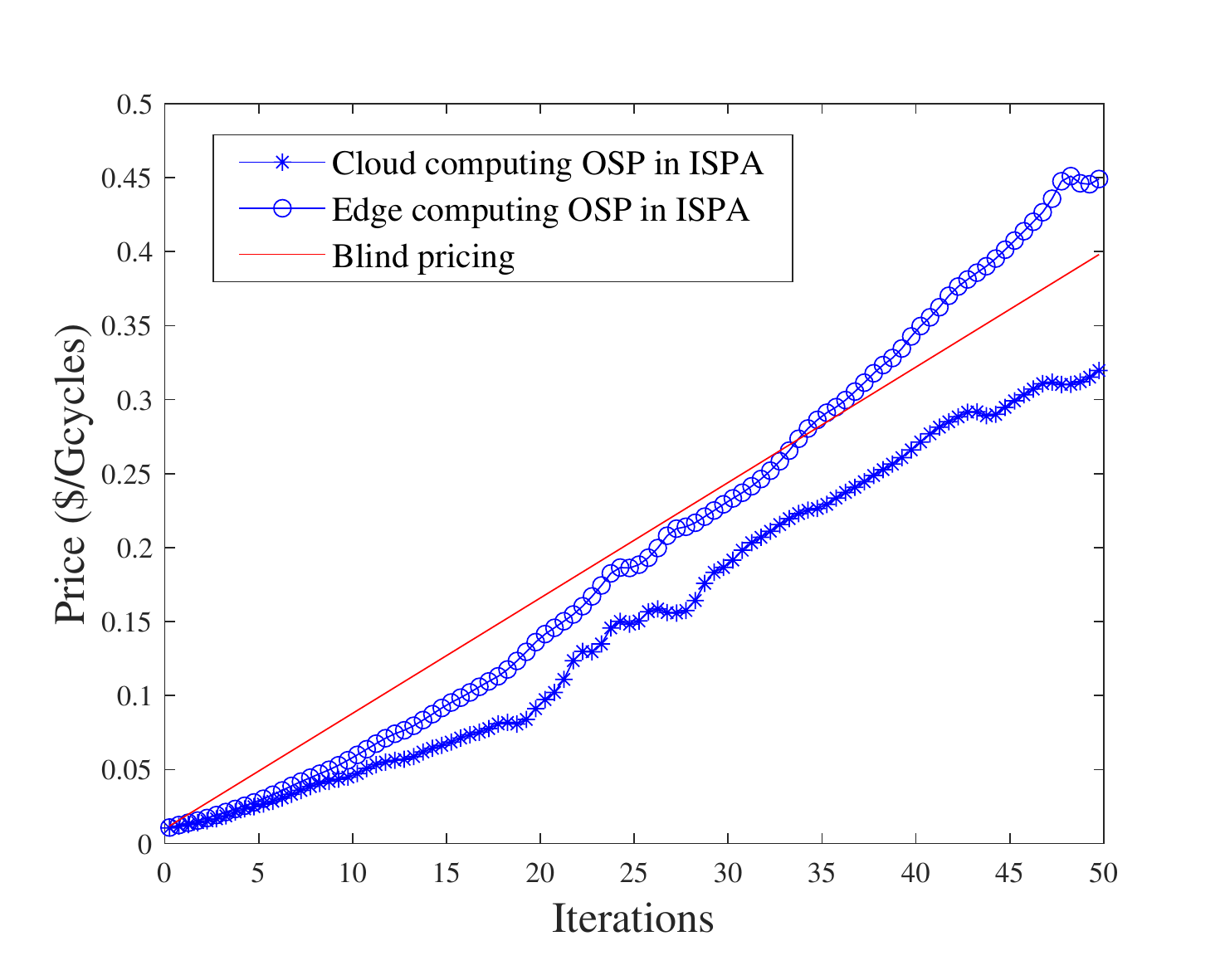}
		\caption{Prices change}
	\end{subfigure}
	\begin{subfigure}[b]{0.4\linewidth}
		\includegraphics[width=\linewidth]{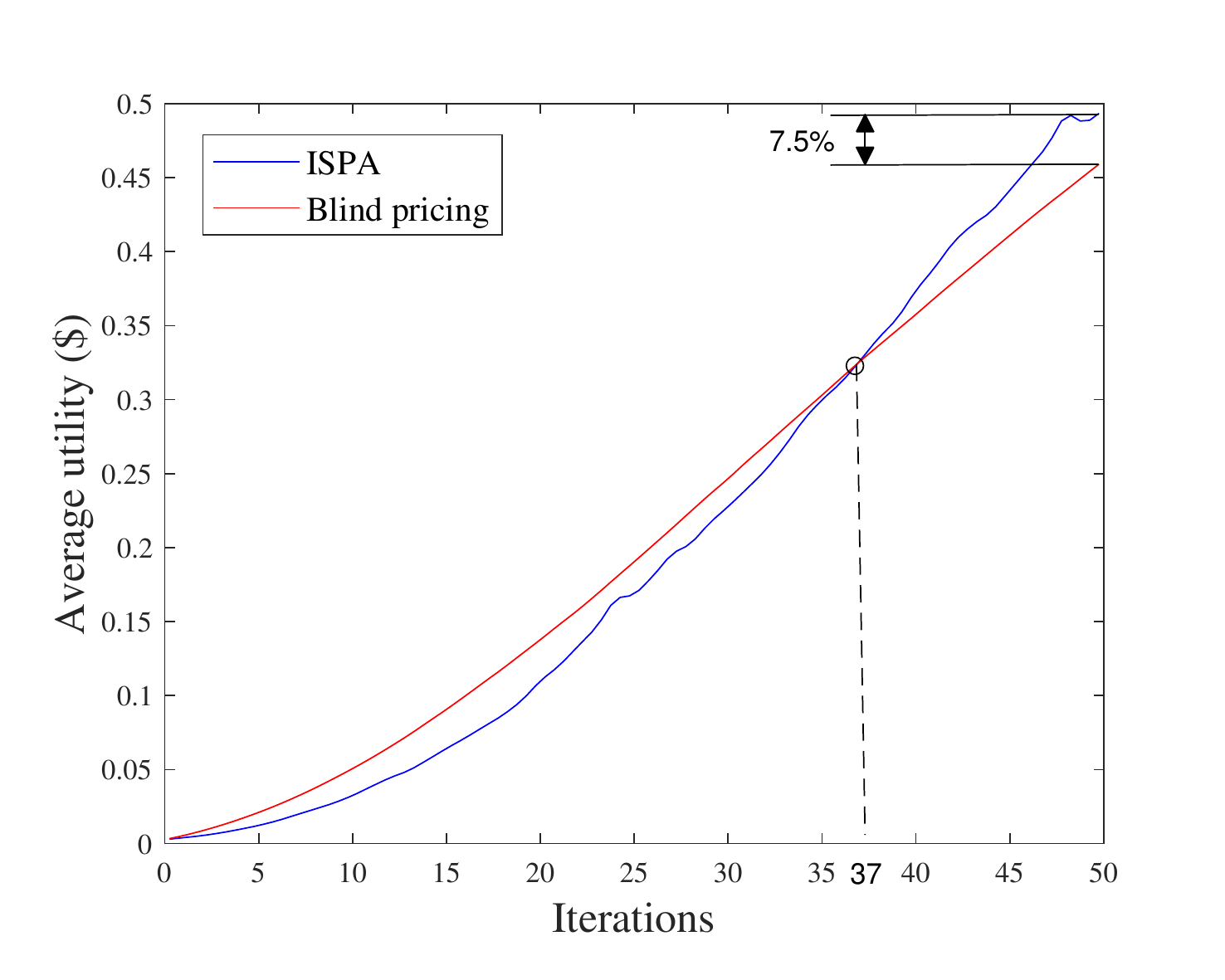}
		\caption{Utilities change}
	\end{subfigure}
	\DeclareGraphicsExtensions.
	\caption{Prices and average utility versus the iterations}
	\label{fig:a2f3}
\end{figure}

\section{Conclusion} \label{sec-7}
This paper provides a novel sdistributed mechanism to address the pricing and task offloading in MEC enabled edge-cloud systems. First, we quantify the QoE of IoT MDs with a disutility function which jointly considers the computing delay, energy consumption and payment cost. Then a multi-leader multi-follower two-tier Stackelberg game model is applied to describe the optimization problems of OSPs and IoT MDs and the existence of SE is analyzed. We first propose IPOA to obtain Nash equilibrium offloading strategies for IoT MDs in the condition that the prices charged by all OSPs are fixed. Furthermore, considering the privacy of IoT MD utility information and the non-cooperative nature of OSPs, we apply backward induction and introduce ISPA so that OSPs can adjust their prices dynamically. Through numerical experiments, results show that our proposed mechanism provides a superior solution to the optimal pricing and task offloading problem in the competitive IoT environment.  

\section*{Acknowledgment}
This work was supported in part by National Natural Science Foundations of China (61821001), YangFan Innovative \& Entrepreneurial Research Team Project of Guangdong Province, Fundamental Research Funds for the Central Universities, and Director Foundation of Beijing Key Laboratory of Work Safety Intelligent Monitoring.

\ifCLASSOPTIONcaptionsoff
  \newpage
\fi

\bibliographystyle{IEEEtran}
\bibliography{IEEEabrv,reference}

\end{document}